\documentclass[final,leqno]{siamltex}

\usepackage{enumitem}
\usepackage{xcolor}
\RequirePackage{graphicx}
\usepackage{epstopdf}
\usepackage{latexsym}
\usepackage{amssymb}
\usepackage{amsmath}
\usepackage{color,soul}

\newtheorem{remark}[theorem]{Remark}

\newtheorem{algorithm}[theorem]{Algorithm}
\newcommand{\PP}{\mathbb P}
\newcommand{\E}{\mathbb E}

\title{Mathematical Analysis of Temperature Accelerated Dynamics}
\author{David Aristoff\thanks{Department of Mathematics, University of Minnesota}\and 
Tony Leli\`evre\thanks{CERMICS, \'Ecole des Ponts ParisTech}}

\begin{document}

\date{May 2013}

\maketitle

\begin{abstract}
We give a mathematical framework for temperature accelerated dynamics
(TAD), an algorithm proposed by M.R. S\o rensen and A.F. Voter
in~\cite{Voter} to efficiently generate metastable stochastic dynamics. 
Using the notion of {\em quasistationary distributions}, we 
propose some modifications to TAD. Then considering the modified 
algorithm in an idealized setting, we show how TAD can be 
made mathematically rigorous.\end{abstract}



\pagestyle{myheadings}
\thispagestyle{plain}
\markboth{D. ARISTOFF AND T. LELI\`EVRE}{MATHEMATICAL ANALYSIS OF TAD}

\section{Introduction}

Consider the stochastic dynamics $X_t$ on ${\mathbb R}^d$ satisfying 
\begin{equation}\label{1}
dX_t = -\nabla V(X_t)\,dt + \sqrt{2\beta^{-1}} \,dW_t, 
\end{equation}
called {\em Brownian dynamics} or {\em overdamped Langevin dynamics}. 
Here $V:{\mathbb R}^d \to \mathbb R$ is a smooth function, $\beta = (k_B T)^{-1}$ 
is a positive constant, and $W_t$ is a standard $d$-dimensional Brownian motion \cite{Oksendal}. 
The dynamics~\eqref{1} is used to model the evolution of the
position vector $X_t$ of $N$ particles (in which case $d=3N$) in an  
energy landscape defined by the potential energy~$V$. This is the 
so-called {\em molecular dynamics}. Typically this energy landscape has 
many metastable states, and in applications it is of interest 
to understand how $X_t$ moves between them. Temperature accelerated 
dynamics (TAD) is an algorithm for computing this {\it metastable 
dynamics} efficiently. (See \cite{Voter} for the original algorithm, 
\cite{Voter3} for some modifications, and 
\cite{Voter2} for an overview of TAD and other 
similar methods for accelerating dynamics.)

Each metastable state corresponds to a basin of attraction $D$  
for the gradient dynamics $dx/dt=-\nabla V(x)$ of a local minimum 
of the potential $V$. In TAD, temperature is raised to force $X_t$ to leave each  
basin more quickly. What would have happened at the original 
low temperature is then extrapolated. To generate metastable 
dynamics of $(X_t)_{t\ge 0}$ at low temperature, this procedure is repeated in each basin. 
This requires the assumptions: 
\begin{itemize}
 \item[(H1)]{$X_t$ immediately reaches local equilibrium upon entering a given basin $D$; and}
\item[(H2)]{An Arrhenius law may be used to extrapolate the exit event at low temperature.}
\end{itemize}
The  Arrhenius (or Eyring-Kramers) law states that, in the small 
temperature regime, the time it takes to transition between neighboring basins $D$ and $D'$ is  
\begin{equation}\label{Arrheniuslaw}
{\nu^{-1}} \exp\left[\frac{|\delta V|}{k_B T}\right],
\end{equation}
where $\delta V$ is the difference in potential energy between the local 
minimum in $D$ and the lowest saddle point along a path joining $D$ to $D'$. Here $\nu$ is a constant (called a {\em prefactor}) depending on the eigenvalues of 
the Hessian of $V$ at the local minimum and at the saddle point, but not on 
the temperature. In practice the Arrhenius law is used when $k_B T \ll |\delta V|$. 
We refer to \cite{Berglund, Bovier, Hanggi, Menz} for details.

TAD is a very popular technique, in particular for
applications in material sciences; {see for example} 
\cite{TAD9, TAD11, TAD6, TAD5, TAD3, TAD10, TAD4, TAD8, TAD1, TAD2, TAD7}.
In this article we provide a mathematical framework for 
TAD, and in particular a mathematical formalism for 
(H1)-(H2). Our analysis will actually concern a slightly modified 
version of TAD. In this modified version, which we call {\em modified TAD}, 
the dynamics is allowed to reach local equilibrium
 after entering a basin, 
thus circumventing assumption~(H1).  
{The 
assumption (H1) is closely related to the 
no recrossings assumption in transition state 
theory; in particular one can see the local 
equilibration steps (modifications (M1) and (M2) below) in modified TAD as a way to account 
for recrossings.}
We note that modified TAD can be used in practice and, since it does not 
require the assumption (H1), may reduce some of the numerical error in 
(the original) TAD. 

To analyze modified TAD, we first 
make the notion of local equilibration precise by using 
{\it quasistationary distributions}, in the spirit of \cite{Tony}, and 
then we circumvent (H2) by introducing an idealized extrapolation procedure 
which is {\it exact}. The result, which we call {\em idealized TAD}, yields exact 
metastable dynamics; see Theorem~\ref{mainthm} below. 
Idealized TAD is not a practical algorithm because it depends on 
quantities related to quasistationary distributions which cannot be efficiently computed. 
However, we show that idealized TAD agrees with modified TAD at low temperature. 
In particular we justify (H2) in modified TAD by showing that at low temperature, 
the extrapolation procedure of idealized TAD agrees with that of
modified TAD (and of TAD), which is based on the Arrhenius law~\eqref{Arrheniuslaw}; 
see Theorem~\ref{theorem2} below.

In this article, we focus on the overdamped Langevin
dynamics~\eqref{1} for simplicity. The algorithm 
is more commonly used in practice
with the Langevin dynamics 
\begin{align}\begin{split}\label{2ndorder}
\begin{cases} dq_t = M^{-1} p_t \,dt\\
dp_t = -\nabla V(q_t)\,dt -\gamma M^{-1}p_t\,dt +\sqrt{2\gamma \beta^{-1}}\,dW_t\end{cases}.
\end{split}
\end{align}
{The notion of quasistationary distributions 
still makes sense for the Langevin dynamics~{\cite{Nier}}, so an extension of our analysis to 
that dynamics 
is in principle possible, though the mathematics 
there are much more difficult due to the degeneracy of the infinitesimal 
generator of {\eqref{2ndorder}}. In particular, some results on the low temperature
asymptotics of the principal eigenvalue and eigenvector for hypoelliptic
diffusions are still missing.}


The paper is organized as follows. In Section~\ref{sec:TAD}, we recall
TAD and present modified TAD. In Section~\ref{sec:idealTAD}, we introduce idealized TAD and 
prove it is exact in terms of metastable dynamics. Finally, in Section~\ref{sec:theta},
we show that idealized TAD and modified TAD are essentially equivalent in the low temperature regime. 
Our analysis in Section~\ref{sec:theta} is restricted to a one-dimensional setting. 
The extension of this to higher dimensions will be the purpose of another work.

Throughout the paper it will be convenient to refer to various objects related to the dynamics~\eqref{1} 
at a high and low temperature, $\beta^{hi}$ and $\beta^{lo}$, as well as at a generic temperature, 
$\beta$. To do so, we use superscripts $^{hi}$ and $^{lo}$ to indicate that we 
are looking at the relevant object at $\beta = \beta^{hi}$ or $\beta = \beta^{lo}$, 
respectively. We drop the superscripts to consider objects at a generic temperature~$\beta$. 

\section{TAD and modified TAD}\label{sec:TAD}

Let $X_t^{lo}$ be a stochastic dynamics obeying~\eqref{1} at 
a low temperature $\beta = \beta^{lo}$, 
and let $S:{\mathbb R}^d \to \mathbb N$ be a function which 
labels the basins of $V$. (So each basin $D$ has the form 
$S^{-1}(i)$ where $i \in \mathbb N$.) The goal of TAD is to efficiently 
estimate the metastable dynamics at low temperature; in other words:
\begin{itemize}
\item{
Efficiently generate a trajectory ${\hat S}(t)_{t\ge 0}$ which has approximately 
the same distribution as  
$S(X_t^{lo})_{t\ge 0}$.}
\end{itemize}
The aim then is to get approximations of {\em trajectories}, {including
distributions of hitting times, time correlations, etc... and
 thus not only
the evolution of the averages of some observables or  
averages of observables with respect to the invariant distribution.} 

At the heart of TAD is the problem of efficiently simulating an  
exit of $X_t^{lo}$ from a generic basin $D$, 
since the metastable dynamics are generated by essentially  
repeating this. To efficiently simulate an exit of $X_t^{lo}$ from $D$, 
{temperature is raised so that} $\beta^{hi} < \beta^{lo}$ 
and a corresponding high temperature dynamics $X_t^{hi}$ is evolved. 
The process $X_t^{hi}$ is allowed to search for various exit 
paths out of $D$ until a stopping time $T_{stop}$; 
each time $X_t^{hi}$ reaches $\partial D$ it is reflected back 
into $D$, the place and time of the attempted exit is recorded, 
and the Arrhenius law~\eqref{Arrheniuslaw} is used to extrapolate a low temperature exit. 
After time $T_{stop}$ the fastest extrapolated low temperature 
exit is selected. This exit is considered an approximation of the 
first exit of $X_t^{lo}$ from $D$. The original algorithm is 
described in Section~\ref{originalTAD} below; a modified 
version is proposed in Section~\ref{modifiedTAD} below.

\subsection{TAD}\label{originalTAD}

In the following, we let $D$ denote a generic basin. We 
let $x_0$ be the minimum of $V$ inside $D$, and we 
assume there are finitely many saddle points, $x_i$ ($i\ge 1$),  
of $V$ on $\partial D$. The original TAD algorithm \cite{Voter} for 
generating the approximate metastable dynamics 
${\hat S}(t)$ is as follows: 
\begin{algorithm}[TAD]\label{alg1} 
Let $X_0^{hi}$ be in the basin $D$, and 
start a low temperature simulation clock 
$T_{tad}$ at zero: $T_{tad} = 0$. Then 
iterate on the visited basins the following:
\begin{enumerate}
\item{Let $T_{sim} = 0$ and $T_{stop} = \infty$. These 
are the simulation and stopping times for the high 
temperature exit search.}
 \item {Evolve $X^{hi}_t$ at $\beta = \beta^{hi}$ starting 
at $t=T_{sim}$  until the 
first time after $T_{sim}$ at which it exits $D$. (Exits 
are detected by checking if the dynamics lands into another 
basin via gradient descent, i.e. the deterministic dynamics $dx/dt = -\nabla V(x)$.)
Call this time $T_{sim}+\tau$.}
\item {Associate a nearby saddle point, $x_i$, of $V$ on $\partial D$ 
to the place where $X^{hi}_t$ exited~$D$. (This 
can be done by using, for example, the nudged elastic 
band method \cite{Henkelman}; see below.)}
\item {Advance the high temperature simulation clock by $\tau$: $T_{sim} = T_{sim} + \tau$.}
\item {If an exit at $x_i$ has already been observed, go to Step 8. If an exit at $x_i$  has not yet been observed, 
set $T_i^{hi} = T_{sim}$ and extrapolate the high temperature exit time to low
temperature using the formula:
\begin{equation}\label{arrhenius}
T_i^{lo} = T_{i}^{hi}\,e^{-(\beta^{hi}-\beta^{lo})(V(x_i)-V(x_0))}. 
\end{equation}
This equation comes from the Arrhenius law~\eqref{Arrheniuslaw} for 
exit rates in the low temperature regime; see the remarks below. }
\item {Update the smallest extrapolated exit time:
\begin{equation*}
T_{min}^{lo} = \min\{T_{min}^{lo}, T_i^{lo}\},
\end{equation*}
and the (index of) the corresponding exit point:
\begin{equation*}
 I_{min}^{lo} = i \hskip10pt\hbox{if}\hskip10pt T_{min}^{lo} = T_i^{lo}.
\end{equation*}
}
\item {Update $T_{stop}$. The stopping time is chosen so that with 
confidence $1-\delta$, an extrapolated low temperature exit time 
smaller than $T_{min}^{lo}$ will not be observed. See equation~\eqref{TADstop} below 
for how this is done.}
\item {If $T_{sim} \le T_{stop}$, reflect $X_t^{hi}$ back into $D$ and
go back to Step 2. Otherwise, proceed to Step~9.}
\item {Set 
\begin{equation*}
 {\hat S}(t) = S(D) \hskip10pt {for} \hskip10pt t \in [T_{tad},T_{tad}+T_{min}^{lo}],
\end{equation*}
and advance the low temperature simulation clock by 
$T_{min}^{lo}$: 
\begin{equation*}
T_{tad} = T_{tad}+T_{min}^{lo}.
\end{equation*}}
\item{Send $X_t^{hi}$ to the new basin, namely the
neighboring basin of $D$ which is attained through the saddle point
$x_{I_{min}^{lo}}$. Then, go back to Step 1, the domain $D$ now being 
the neighboring basin.}
\end{enumerate}
\end{algorithm}

The nudged elastic band method \cite{Henkelman} consists, starting from a 
trajectory leaving $D$, of computing by a gradient descent method the closest 
minimum energy path leaving $D$, with the end points of the trajectory being fixed. 
This minimum energy path necessarily leaves $D$ through a saddle point.

\begin{remark}
{ When the overdamped Langevin dynamics leaves
a basin near a saddle point, its first re-entrance 
into that basin is immediate. 
Thus, Algorithm~{\ref{alg1}} does not really make 
sense for overdamped Langevin dynamics. (With the Langevin dynamics~{\eqref{2ndorder}}, however, this difficulty does not arise.) 
In modified TAD, defined below, we will allow the 
dynamics to evolve away from the boundary of a basin 
after an exit event, thus circumventing this problem.}
\end{remark}

Below we comment on the equation~\eqref{arrhenius} from 
which low temperature exit times are extrapolated,  
as well as the stopping time $T_{stop}$. 

\begin{itemize}
\item {\bf Low temperature extrapolation}. 

The original 
TAD uses the following kinetic Monte Carlo (KMC) framework~\cite{KMC}. 
For a given basin $D$, it is assumed that the time ${\tilde T}_i$ 
to exit through the saddle point $x_i$ of $V$ on $\partial D$ is exponentially 
distributed with rate $\kappa_i$ given by the 
Arrhenius law~\eqref{Arrheniuslaw}: 
\begin{equation}\label{expparam}
 \kappa_i \equiv {\nu_i} e^{-\beta (V(x_i)-V(x_0))}
\end{equation}
where we recall $\nu_i$ is a temperature independent 
prefactor and $x_0$ is the minimum of $V$ in $D$. 
An exit event from $D$ at temperature $\beta$ is obtained 
by sampling independently the times ${\tilde T}_i$ for 
all the saddle points $x_i$ on $\partial D$, 
then selecting the smallest time and the 
corresponding saddle point.

In TAD, this KMC framework is used for both 
temperatures $\beta^{lo}$ and $\beta^{hi}$. 
That is, it is assumed that the high and low temperature 
exit times ${\tilde T}_i^{hi}$ and ${\tilde T}_i^{lo}$ 
through each saddle point $x_i$ satisfy:
\begin{align}\begin{split}\label{explaw}
\PP({\tilde T}_i^{hi} > t) &= e^{-{\kappa}_i^{hi} t}\\
\PP({\tilde T}_i^{lo} > t) &= e^{-{\kappa}_i^{lo} t}
\end{split}
\end{align}
where
\begin{align}\begin{split}\label{prefactors}
\kappa_i^{hi} &= {\nu_i} e^{-\beta^{hi}(V(x_i)-V(x_0))}\\
\kappa_i^{lo} &= {\nu_i} e^{-\beta^{lo}(V(x_i)-V(x_0))}
\end{split}
\end{align}
Observe that then  
\begin{equation*}
{\tilde T}_i^{hi}\, {\frac{{\kappa}_i^{hi}}{{\kappa}_i^{lo}}} = {\tilde T}_i^{hi} \,e^{-(\beta^{hi}-\beta^{lo})(V(x_i)-V(x_0))}
\end{equation*}
has the same probability law as ${\tilde T}_i^{lo}$. This leads to the 
extrapolation formula~\eqref{arrhenius}.

The assumption of exponentially distributed exit times $T_i^{hi}$ 
and $T_i^{lo}$ is 
valid only if the dynamics at both temperatures immediately reach local equilibrium upon 
entering a basin; see (H1) and Theorem~\ref{theorem0a} below. In modified TAD, described below, 
we circumvent this immediate equilibration assumption by allowing the dynamics 
at both temperatures to 
reach local equilibrium. In particular, in modified TAD the low temperature assumption 
is no longer needed to get exponential exit distributions as in~\eqref{explaw}. 
On the other hand, to get the {\em rate constants} in~\eqref{prefactors} 
-- and by extension the extrapolation rule~\eqref{arrhenius}; see (H2) -- a low 
temperature assumption is required. {We will justify both~{\eqref{explaw}} and~{\eqref{prefactors}} 
in the context of modified TAD. More precisely 
we show that~{\eqref{explaw}} will be valid at any temperature, 
while a low temperature assumption is needed to justify~{\eqref{prefactors}}. 
Note that, inspecting equation~{\eqref{prefactors}}, the 
low temperature assumption will be required for {\em both} temperatures 
used in TAD -- so $1/\beta^{hi}$
will be small in an absolute sense, but 
large compared to $1/\beta^{lo}$.}

\item{\bf Stopping time.}

The stopping time $T_{stop}$ is chosen so that if the high temperature 
exit search is stopped at time $T_{stop}$, then with probability 
$1-\delta$, the smallest extrapolated low temperature exit time will be 
correct. Here $\delta$ is a user-specified parameter.

To obtain a formula for the 
stopping time $T_{stop}$ it is assumed that, in addition to (H1)-(H2):
\begin{itemize}
 \item[(H3)] {There 
is a minimum, $\nu_{min}$, to all the prefactors in
equation~{\eqref{prefactors}}}: $$\forall i \in \{1,\ldots k\}, \nu_i
\ge \nu_{min},$$ 
\end{itemize}
where $k$ denotes the number of saddle points on $\partial D$.

Let us now explain how this assumption is used to determine $T_{stop}$.
Let $T$ be a deterministic time. If a high temperature first exit time through $x_i$, $T_i^{hi} > T$, 
extrapolates to a low temperature time less than $T_{min}^{lo}$, then from~\eqref{arrhenius}, 
\begin{equation*}
 V(x_i)-V(x_0) \le \frac{\log(T_{min}^{lo}/T)}{\beta^{lo}-\beta^{hi}}
\end{equation*}
and so
\begin{equation}\label{ki}
{\kappa}_i^{hi} = \nu_i e^{-\beta^{hi}(V(x_i)-V(x_0))} \ge \nu_{min}\exp\left(\frac{\beta^{hi}\log(T_{min}^{lo}/T)}{\beta^{hi}-\beta^{lo}}\right).
\end{equation}
In TAD it is required that this event has a low probability 
$\delta$ of occurring, that is, 
\begin{equation}\label{delta}
\PP(T_i^{hi} > T) = e^{-{\kappa}_i^{hi} T} < \delta.
\end{equation}
Using~\eqref{ki} in~\eqref{delta}, one sees that it suffices that 
\begin{equation*}
\exp\left[-\nu_{min}\exp\left(\frac{\beta^{hi}\log(T_{min}^{lo}/T)}{\beta^{hi}-\beta^{lo}}\right)T\right] < \delta.
\end{equation*}
Solving this inequality for $T$, one obtains 
\begin{equation}
 T > \frac{\log(1/\delta)}{\nu_{min}}\left(\frac{\nu_{min}T_{min}^{lo}}{\log(1/\delta)}\right)^{\beta^{hi}/\beta^{lo}}.
\end{equation}
The stopping time $T_{stop}$ is then chosen to be the right hand side of the above:
\begin{equation}\label{TADstop}
 T_{stop} \equiv \frac{\log(1/\delta)}{\nu_{min}}\left(\frac{\nu_{min}T_{min}^{lo}}{\log(1/\delta)}\right)^{\beta^{hi}/\beta^{lo}}.
\end{equation}
(It is calculated using the current value of $T_{min}^{lo}$.)
The above calculation shows that at simulation time 
$T_{stop}$, with probability at least $1-\delta$, $T_{min}^{lo}$ 
is the same as the smallest extrapolated 
low temperature exit time which would have been observed 
with no stopping criterion. 

For TAD to be practical, the stopping time $T_{stop}$ must be (on average) 
smaller than the exit times at low temperature. The stopping 
time of course depends on the choice of $\nu_{min}$ and $\delta$. 
In practice a reasonable value for $\nu_{min}$ may be known a priori 
\cite{Voter} or obtained by a crude approximation \cite{Voter2}. 
For a given $\delta$, if too large a value of $\nu_{min}$ is used, 
the low temperature extrapolated times 
may be incorrect with probability greater than $\delta$. 
On the other hand, if the value of $\nu_{min}$ is too small,  
then the extrapolated times will be correct with probability $1-\delta$, but 
computational efficiency will be compromised. The usefulness of TAD comes 
from the fact that, in practice, $\nu_{min}$ and $\delta$ can often 
be chosen such that the correct low temperature exit event 
is found by time $T_{stop}$ with large probability $1-\delta$, 
{\it and} $T_{stop}$ is on average much smaller than the exit times 
which would be expected at low temperature. In practical applications, 
TAD has provided simulation time scale boosts of up to $10^9$~\cite{TAD4}.
\end{itemize}

\begin{remark} 
One alternative to TAD is a brute force saddle point 
search method, in which one evolves the 
system at a high temperature $\beta^{hi}$ to 
locate saddle points of $V$ on $\partial D$. 
{(There are other popular techniques 
in the literature to locate saddle points, 
many of which do not use high 
temperature or dynamics; see 
for example~{\cite{Mousseau}}.)}
{Once one is confident that all the physically relevant 
saddle points are found}, the times ${\tilde T}_i^{lo}$ 
to exit through each $x_i$ at low temperature 
can be directly sampled from exponential distributions with 
parameters $\kappa_i$ as in~\eqref{expparam}, using $\beta \equiv \beta^{lo}$. 
(Estimates are available for the $\nu_i$ at low temperature; they 
depend on the values of $V$ and the Hessian matrix of $V$ at $x_i$ and $x_0$. 
See for example \cite{Bovier}.) 

The advantage of TAD over 
a brute force saddle point search method is that 
in TAD, there is a well-defined stopping criterion 
for the saddle point search at temperature $\beta^{hi}$, in 
the sense that the saddle point corresponding to 
the correct exit event at temperature $\beta^{lo}$ 
will be obtained with a user-specified probability. 
In particular, TAD does not require all the saddle points to be found. 
\end{remark}

\subsection{Modified TAD}\label{modifiedTAD}
Below we consider some modifications, (M1)-(M3), to TAD, calling 
the result {\it modified TAD}. The main modifications, (M1)-(M2) below, 
will ensure that the exponential rates assumed 
in TAD are justified. We also introduce a different stopping 
time, (M3). (See the discussion below Algorithm~\ref{alg2}.) 
We note that some of these features are currently being used by 
practitioners of TAD~\cite{Voterpriv}. Here are the three modifications: 

\begin{itemize}
 \item[(M1)] We include a decorrelation step in which an underlying 
low temperature dynamics $(X_t^{lo})_{t\ge 0}$ 
finds local equilibrium in some basin $D$ before we start searching for exit pathways at high 
temperature;

\item[(M2)] Before searching for exit pathways 
out of $D$, we sample local equilibrium at high temperature 
in the current basin $D$, without advancing any clock time;

\item[(M3)] We replace the stopping time~\eqref{TADstop} with 
\begin{equation}\label{modstop}
T_{stop} = T_{min}^{lo}/C,
\end{equation}
where $C$ is a lower bound of the 
minimum of $e^{-(\beta^{hi}-\beta^{lo})(V(x_i)-V(x_0))}$ 
over all the saddle points, $x_i$, of $V$ on $\partial D$. 
\end{itemize}
\begin{remark} 
{In (M3) above we are assuming some a priori knowledge of 
the system, in particular a lower bound of the energy 
barriers $V(x_i)-V(x_0)$, $i \in
\{1,\ldots k\}$. Such a lower bound will not be known in every  situation, 
but in some cases, practitioners can obtain such a bound, see for example~{\cite{Voter3}}. See also the 
discussion in the section ``Stopping time'' below.}
\end{remark}

The modified algorithm is 
as follows; for the reader's convenience we have boxed off the steps of 
modified TAD which are different from TAD.

\begin{algorithm}[Modified TAD]\label{alg2}
Let $X_0^{lo}$ be in the basin $D$, 
set a low temperature simulation clock $T_{tad}$ to zero:
$T_{tad} = 0$, and choose a (basin-dependent) decorrelation time $T_{corr}>0$. 
Then iterate on the visited basins the following:
\vskip5pt

\fcolorbox{black}[HTML]{E9F0E9}{\parbox{\textwidth}{
\begin{enumerate}

\item[]{\bf Decorrelation step:}
 \item{Starting at time $t = T_{tad}$, evolve $X_t^{lo}$ at temperature $\beta = \beta^{lo}$ 
according to~\eqref{1} in the current basin $D$.}
\item{If $X_t^{lo}$ exits $D$ at a time 
$T_{tad} + \tau < T_{tad} + T_{corr}$, 
then set 
\begin{equation*}
{\hat S}(t) = S(D), \hskip10pt t \in [T_{tad},T_{tad}+\tau],
\end{equation*}
advance the low temperature clock by $\tau$:
$T_{tad} = T_{tad} + \tau$,
then go back to Step 1, where $D$ is now the 
new basin. Otherwise, set
\begin{equation*}
{\hat S}(t) = S(D), \hskip10pt t \in [T_{tad},T_{tad}+T_{corr}],
\end{equation*}
advance the low temperature clock by $T_{corr}$: $T_{tad} = T_{tad} + T_{corr}$,  
and initialize the exit step by setting $T_{sim} = 0$ and $T_{stop} = \infty$. 
Then proceed to the exit step.}  
\end{enumerate}}}
\begin{enumerate}[leftmargin=0.79in]
\item[]{\bf Exit step:}
\end{enumerate}

\fcolorbox{black}[HTML]{E9F0E9}{\parbox{\textwidth}{
\begin{enumerate}
\item[1.]{Let $X_{T_{sim}}^{hi}$ be a sample of the 
dynamics~\eqref{1} in local equilibrium in $D$ at 
temperature $\beta = \beta^{hi}$. 
See the remarks below for how this sampling is done. None 
of the clocks are advanced in this step.}
\end{enumerate}}}

\begin{enumerate}[leftmargin=0.79in]
\item[2.] {Evolve $X^{hi}_t$ at $\beta = \beta^{hi}$ starting 
at $t=T_{sim}$  until the 
first time after $T_{sim}$ at which it exits $D$. 
Call this time $T_{sim}+\tau$.}
\item[3.] {Using the nudged elastic band method, 
associate a nearby saddle point, $x_i$, of $V$ on $\partial D$ 
to the place where $X^{hi}_t$ exited $D$.}
\item[4.] {Advance the simulation clock by $\tau$: $T_{sim} = T_{sim} + \tau$.}
\item[5.] {If an exit at $x_i$ has already been observed, go to Step 8. 
If an exit at $x_i$ has not yet been observed, 
set $T_i^{hi} = T_{sim}$ and
\begin{equation}\label{arrhenius2}
T_i^{lo} = T_{i}^{hi}\,e^{-(\beta^{hi}-\beta^{lo})(V(x_i)-V(x_0))}. 
\end{equation} }
\item[6.] {Update the lowest extrapolated exit time:
\begin{equation*}
T_{min}^{lo} = \min\{T_{min}^{lo}, T_i^{lo}\},
\end{equation*}
and the (index of) the corresponding exit point:
\begin{equation*}
 I_{min}^{lo} = i \hskip10pt\hbox{if}\hskip10pt T_{min}^{lo} = T_i^{lo}.
\end{equation*}
}
\end{enumerate}

\fcolorbox{black}[HTML]{E9F0E9}{\parbox{\textwidth}{
\begin{enumerate}
\item[7.]{Update $T_{stop}$: 
\begin{equation}\label{modstop2}
T_{stop} = T_{min}^{lo}/C,
\end{equation}
where $C$ is a lower bound of 
the minimum of $e^{-(\beta^{hi}-\beta^{lo})(V(x_i)-V(x_0))}$ 
over all the saddle points, $x_i$, of $V$ on $\partial D$.}
\item[8.]{If $T_{sim} \le T_{stop}$, 
go back to Step 1 of the exit step; otherwise, proceed to Step~9. }
\end{enumerate}}}

\begin{enumerate}[leftmargin=0.79in]
\item[9.]{Set 
\begin{equation*}
 {\hat S}(t) = S(D) \hskip10pt {for} \hskip10pt t \in [T_{tad},T_{tad}+T_{min}^{lo}],
\end{equation*}
and advance the low temperature simulation clock by 
$T_{min}^{lo}$: 
\begin{equation*}
T_{tad} = T_{tad}+T_{min}^{lo}.
\end{equation*}}
\end{enumerate}

\fcolorbox{black}[HTML]{E9F0E9}{\parbox{\textwidth}{
\begin{enumerate}
\item[10.]{Set {$X_{T_{tad}}^{lo} = X_{T_{I}^{hi}}^{hi}$ 
where $I \equiv I_{min}^{lo}$}. 
Then go back to the decorrelation step, the domain $D$ now being the neighboring basin, 
namely the one obtained by exiting through {$X_{T_{I}^{hi}}^{hi}$}.}
\end{enumerate}}}

\end{algorithm}
\vskip10pt

\begin{itemize}
\item{\bf Local equilibrium in $D$: (M1) and (M2).} 

We introduce the decorrelation step -- see (M1) -- in order to 
ensure that the low temperature dynamics reaches local equilibrium 
in $D$. Indeed, for sufficiently large $T_{corr}$ the low temperature 
dynamics reaches local equilibrium in some
basin.  The convergence to local equilibrium will be made precise in
Section~\ref{sec:idealTAD} using the notion of the {\em quasistationary
  distribution}.  See also~\cite{Gideon,Tony}, in particular for a discussion of the choice of $T_{corr}$. 
Local equilibrium will in general be reached at different times in
different basins, so we allow $T_{corr}$ to be basin dependent.
We note that a similar decorrelation step is used in another 
accelerated dynamics proposed by A.F. Voter, the Parallel Replica Dynamics~\cite{VoterParRep}. {The decorrelation step accounts for 
barrier recrossing events: the dynamics is allowed to evolve 
exactly at low temperature after the exit step, capturing any 
possible barrier recrossings, until local equilibrium is 
reached in one of the basins.}

The counterpart of the addition of this decorrelation step is that, 
from (M2), in the exit step we also start the high temperature dynamics from 
local equilibrium in the current basin $D$. 
{A similar step is actually being used by current practitioners 
of TAD~{\cite{Voterpriv}}, though this step is not mentioned 
in the original algorithm~{\cite{Voter}}.} 
To sample local equilibrium in $D$, one can for example take the 
end position of a a sufficiently 
long trajectory of~\eqref{1} which does not exit $D$. 
See \cite{Tony, Gideon} for some algorithms to efficiently sample 
local equilibrium; we remark that this is expected to become more computationally demanding 
as temperature increases. 

To extrapolate the exit event at low temperature from the exit events at high temperature, 
we need the dynamics at both temperatures to be in local equilibrium. We note that 
the changes (M1)-(M2) in modified TAD are actually a practical way to get rid of the 
error associated with the assumption (H1) in TAD.

\item{\bf Stopping time: (M3).} 

In (M3) we introduce a stopping $T_{stop}$ 
such that, with probability $1$, the shortest extrapolated 
low temperature exit time is found by time $T_{stop}$. (Recall that with the 
stopping time of TAD, we have only a confidence level $1-\delta$.) 

Note that for the stopping time $T_{stop}$ to be implemented in 
~\eqref{modstop2}, we need some a priori knowledge about energy 
barriers, in particular a lower bound $E_{min}>0$ for all the 
differences $V(x_i)-V(x_0)$, where $x_i$ ranges over the saddle 
points on the boundary of a given basin:
\begin{itemize}
 \item[(H3')] {There 
is a minimum, $E_{min}$, to all the energy barriers}: $$\forall i \in
\{1,\ldots k\}, V(x_i) - V(x_0) \ge E_{min}.$$ 
\end{itemize}
 {If a lower bound $E_{min}$ is known}, then
we can choose $C$ accordingly so that in equation~\eqref{modstop2} we obtain 
\begin{equation}\label{modifiedstop}
 T_{stop} = T_{min}^{lo} e^{(\beta^{hi}-\beta^{lo})E_{min}}.
\end{equation}
A simple computation then shows that {under assumption (H3')}, any high temperature exit 
time occurring after $T_{stop}$ cannot extrapolate to a low temperature 
exit time smaller than $T_{min}^{lo}$. To see that~\eqref{modifiedstop} leads 
to an efficient algorithm, recall that TAD is expected to be correct only in the regime where 
$\beta^{hi} \gg E_{min}$, which since $\beta^{hi}\ll\beta^{lo}$ means the 
exponential in~\eqref{modifiedstop} should be very small.

As the computational savings of TAD comes from the
fact that the simulation time of the exit step, namely $T_{stop}$, is
much smaller than the exit time that would have been observed at low
temperature, the choice of stopping time in TAD is of critical importance. 
Both of the stopping times~\eqref{TADstop} and~\eqref{modstop}
are used in practice; see {{\cite{Voter3}}} for a presentation of TAD 
with the stopping formula~\eqref{modstop}, and 
\cite{Montalenti2} for an application. {The original stopping
time~{\eqref{TADstop}} requires a lower 
bound for the prefactors in the Arrhenius law~{\eqref{prefactors}} (see
assumption (H3) above, in the remarks following
Algorithm~{\ref{alg1}}). The stopping time~{\eqref{modstop}} requires an 
assumption on the minimum energy barriers; see assumption (H3') above}. The formula~\eqref{modstop} 
may be preferable in case minimum energy barriers are known, 
since it is known to scale 
better with system size than \eqref{TADstop}. The formula~\eqref{TADstop} is 
advantageous if minimum energy barriers are unknown but a reasonable 
lower bound for the minimum prefactor $\nu_{min}$ is available.

{We have chosen  
the stopping time~{\eqref{modstop}} instead of~{\eqref{TADstop}} 
mostly for mathematical convenience -- in particular so that in our 
Section~{\ref{sec:idealTAD}} analysis we do not have the error $\delta$ 
associated with~{\eqref{TADstop}}. A similar  
analysis can be done under assumption (H3) with the stopping time~{\eqref{TADstop}}, 
modulo the error $\delta$.}
\end{itemize}

We comment that modified TAD is an algorithm which can be implemented in 
practice, and which circumvents the error in the original TAD 
arising from the assumption (H1).

\section{Idealized TAD and mathematical analysis}\label{sec:idealTAD}

In this section we show that under certain idealizing assumptions, 
namely (I1)-(I3) and (A1) below, modified TAD is {\it exact} in the sense that the 
simulated metastable dynamics ${\hat S}(t)_{t\ge 0}$ 
has the same law as the true low temperature metastable dynamics $S(X_t^{lo})_{t\ge 0}$. 
We call this idealization of modified TAD {\it idealized TAD}. 
Our analysis will show that idealized TAD and modified TAD 
agree in the limit $\beta^{hi},\beta^{lo} \to \infty$ and $T_{corr} \to \infty$. 
Since idealized TAD is exact, it follows that modified TAD is exact in the 
limit $\beta^{hi},\beta^{lo} \to \infty$ and $T_{corr} \to \infty$. 

In idealized TAD, we assume that at the end of the decorrelation step and 
at the start of the exit step of modified TAD,  
we are in {\it exact} local equilibrium; see (A1) and (I1). We formalize this using the 
notion of quasistationary distributions, defined below. We also assume that the way in which 
we exit near a given saddle point $x_i$ in the exit step does not affect 
the metastable dynamics in the decorrelation step; see (I2). The remaining  
idealization, whose relation to modified TAD is maybe not so clear at
first sight, is to replace the exponential  
$\exp[-(\beta^{hi}-\beta^{lo})(V(x_i)-V(x_0))]$ of~\eqref{arrhenius} 
with a certain quantity $\Theta_i$ depending on 
the flux of the quasistationary distribution across $\partial D$; see~(I3).  
In Section~\ref{sec:theta} we justify this by showing that the two 
agree asymptotically as $\beta^{hi},\beta^{lo} \to \infty$ in a one-dimensional setting. 

\subsection{Notation and quasistationary distribution}\label{section2b}
Here and throughout, $D$ is an (open) domain with $C^2$ boundary $\partial D$ 
and $X_t^x$ is a stochastic process evolving according to~\eqref{1} starting at 
$X_0^x = x$ (we suppress the superscript where it is not needed). We  
write $\PP(\cdot)$ and $\E[\cdot]$ for various probabilities and expectations,  
the meaning of which will be clear from context. We write 
$Y \sim \mu$ for a random variable sampled from the probability 
measure $\mu$ and $Y\sim {\cal E}(\alpha)$ for an 
exponentially distributed random variable with parameter $\alpha$. 

Recalling the notation 
of Section~\ref{sec:TAD}, we assume that $\partial D$ is partitioned 
into $k$ (Lebesgue measurable) subsets $\partial D_i$ containing the saddle points $x_i$ of $V$, 
$i=1,\ldots,k$ (see Fig~\ref{fig1}):
\begin{equation*}
\partial D = \cup_{i=1}^k \partial D_i \hskip10pt \hbox{ and } 
\hskip10pt \partial D_i \cap \partial D_j = \emptyset \hbox{ if } i \neq j.
\end{equation*}
We assume that any exit through $\partial D_i$ is associated to the
saddle point $x_i$ in Step 3 of TAD. In other
words, $\partial D_i$ corresponds to the basin of attraction of the
saddle point $x_i$ for the nudged elastic band method. 

\begin{figure}
\begin{center}
\includegraphics[scale=0.5]{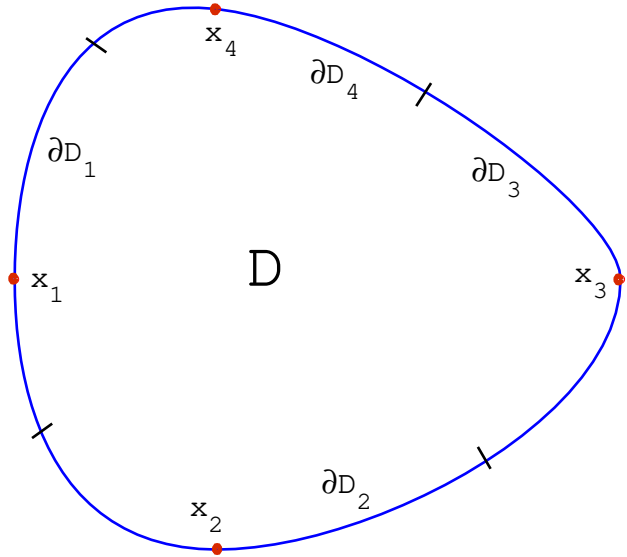}
\end{center}
\caption{The domain $D$ with boundary partitioned into $\partial
  D_1,\ldots,\partial D_4$ (here $k=4$) 
by the black line segments. $V$ has exactly one saddle point in each
$\partial D_i$, located at $x_i$. } 
\label{fig1}
\end{figure}

Essential to the analysis below will be the notion of {\it quasistationary distribution}, 
which we define below, recalling some facts which will be needed in our analysis. Consider the 
infinitesimal generator of~\eqref{1}: 
\begin{equation*}
L = -\nabla V \cdot \nabla +\beta^{-1} \Delta, 
\end{equation*}
and let $(u,-\lambda)$ be the principal eigenvector/eigenvalue pair for $L$ with 
homogeneous Dirichlet (absorbing) boundary conditions on $\partial D$:
\begin{equation}\label{eq:ulambda}
\left\{
\begin{aligned}
Lu&=-\lambda u \text{ in } D, \\
u&=0 \text{ on } \partial D.
\end{aligned}
\right.
\end{equation}
It is known (see \cite{Tony})
that $u$ is signed and $\lambda >0$; we choose $u>0$ and for the moment do not specify a normalization. 
Define a probability measure $\nu$ on $D$ by 
\begin{equation}\label{qsd}
d\nu = \frac{u(x) e^{-\beta V(x)}\,dx}{\int_D u(x) e^{-\beta V(x)}\,dx}.
\end{equation}
The measure $\nu$ is called the {\it quasistationary distribution} (QSD) on $D$; 
the name comes from the fact that $\nu$ has 
the following property: for $(X_t)_{t \ge 0}$ a solution to~\eqref{1},
starting from any distribution with support in $D$,
\begin{equation}\label{eq:local_eq}
\nu(A) = \lim_{t\to \infty} \PP(X_t \in A\,\big|\, X_s \in D,\,0\le s
\le t) \hskip20pt \hbox{for any measurable set }A\subset D.
\end{equation}
The following is proved in \cite{Tony}, and will be essential for our results:
\begin{theorem}\label{theorem0a}
Let $X_t$ be a solution to~\eqref{1} with $X_0 \sim \nu$, and let 
\begin{equation*}
\tau = \inf\{t>0\,:\,X_t \notin D\}
\end{equation*}
Then: (i) $\tau \sim {\cal E}(\lambda)$ and (ii) $\tau$ and $X_\tau$ are independent.
\end{theorem}
We will also need the following formula from \cite{Tony} for the exit point distribution:
\begin{theorem}\label{theorem0b}
Let $X_t$ and $\tau$ be as in Theorem~\ref{theorem0a}, and let 
$\sigma_{\partial D}$ be Lebesgue measure on $\partial D$. The measure $\rho$ on  
$\partial D$ defined by  
\begin{equation}\label{eq:rho}
 d\rho = -\frac{\partial_n\left(u(x) e^{-\beta V(x)}\right)\,d\sigma_{\partial D}}{\beta\lambda\int_D u(x) e^{-\beta V(x)}\,dx}
\end{equation}
is a probability measure, and for any measurable $A \subset \partial D$, 
\begin{equation*}
 \PP(X_\tau \in A) = \rho(A).
\end{equation*}
\end{theorem}

As a corollary of these two results we have the following, which will be central to our analysis:
\begin{corollary}\label{corollary1}
Let $X_t$, $\tau$ and $\rho$ be as in Theorems~\ref{theorem0a}-\ref{theorem0b}, 
and define 
\begin{equation}\label{eq:pi}
p_i =\rho(\partial D_i)
\end{equation}
to be the exit probability through $\partial D_i$. Let $I$ be the
discrete random variable defined by: for $i = 1, \ldots, k$,
$$I=i \text{ if and only if } X_\tau \in \partial D_i.$$
Then (i) $\tau \sim {\cal E}(\lambda)$, (ii) $\PP(I=i) = p_i$, and 
(iii) $\tau$ and $I$ are independent.
\end{corollary}

Throughout we omit the dependence of $\lambda$, $\nu$, and $\rho$ on 
the basin $D$; it should be understood from context.

\begin{remark} {We assume that $D$ has $C^2$ 
boundary so that standard elliptic regularity results and trace
theorems give a meaning to the formula~{\eqref{eq:rho}} used to define $\rho$ 
in Theorem~{\ref{theorem0b}}. For basins of attraction 
this assumption will not be satisfied, as 
the basins will have 
``corners''. This is actually a minor technical point. The probability
measure $\rho$ can be defined for any Lipschitz domain 
$D$ using the
following two steps: first, $\rho$ can be defined in
$H^{-1/2}(\partial \Omega)$ using the definition (equivalent
to~{\eqref{eq:rho}}): for any $v \in H^{1/2} (\partial D)$}
$$\langle v , d\rho \rangle =\frac{\int_D  ( - \beta^{-1} \nabla w \cdot \nabla  u + \lambda w u)  \exp(-\beta V) }{\lambda \int_{D} u \exp(-\beta V)}
$$
{where $w \in H^1(D)$ is any lifting of $v$ ($w|_{\partial D}=v$). Second, it is easy to check
that $\rho$ actually defines a} {\em non-negative} {distribution on
$\partial D$, for example by using as a lifting the solution to}
$$
\left\{
\begin{aligned}
L w &= 0  \text{ in } D,\\
w&=v \text{ on } \partial D,
\end{aligned}
\right.
$$
{since, by the maximum principle, $w \ge 0$, and then, 
$\langle v , d\rho \rangle
= \frac{\int_D  
    \lambda w u  \exp(-\beta V) }{\lambda \int_{D} u \exp(-\beta
    V)}$. One finally concludes using a Riesz representation theorem
  due to Schwartz:
  any non-negative distribution with total mass one defines a
  probability measure.}
\end{remark}
\subsection{Idealized TAD}\label{section3}

In this section we consider an idealized version of 
modified TAD, which we call 
{\it idealized TAD}. The idealizations, (I1)-(I3) below,  
are introduced so that the algorithm 
can be rigorously analyzed using the mathematical 
formalisms in Section~\ref{section2b}. 

\begin{itemize}
\item[(I1)] At the start of the exit step, the high temperature dynamics is 
initially distributed according to the QSD in $D$: $X_{T_{sim}}^{hi} \sim \nu^{hi}$;
\item[(I2)] At the end of the exit step, the extrapolated low 
temperature exit point $X_{T_{tad}}^{lo}$ 
is sampled exactly from the conditional exit point distribution in 
$\partial D_{I_{min}^{lo}}$ at low temperature: 
\begin{equation}\label{exactexit}
X_{T_{tad}}^{lo} \sim \left[\rho^{lo}\left(\partial D_{I_{min}^{lo}}\right)\right]^{-1} \rho^{lo}|_{\partial D_{I_{min}^{lo}}}
\end{equation}
\item[(I3)]In the exit step, the quantity 
\begin{equation*}
e^{-(\beta^{hi}-\beta^{lo})(V(x_i)-V(x_0))}
\end{equation*}
is everywhere replaced by 
\begin{equation}\label{ratios}
 \Theta_i \equiv \frac{\lambda^{hi} p_i^{hi}}{\lambda^{lo} p_i^{lo}},
\end{equation}
where, as in~\eqref{eq:pi}, $p_i^{lo}=\rho^{lo}(\partial
D_i)$ and $p_i^{hi}=\rho^{hi}(\partial D_i)$. Thus, the 
extrapolation equation~\eqref{arrhenius} 
is replaced by 
 \begin{equation}\label{extrapolate}
T_i^{lo} = T_{i}^{hi}\Theta_i
\end{equation}
and the formula for updating $T_{stop}$ is:
\begin{equation}\label{stop}
 T_{stop} = T_{min}^{lo}/C
\end{equation}
where $C$ is chosen so that $C\le \min_{1\le i\le k} \Theta_i$. 
\end{itemize}

We state idealized TAD below as an ``algorithm'', even 
though it is not practical: in general we cannot exactly sample 
$\nu^{hi}$ or the exit distributions  
$\left[\rho^{lo}\left(\partial D_{i}^{lo}\right)\right]^{-1} \rho^{lo}|_{\partial D_{i}^{lo}}$,
and the quantities $\Theta_i$ are not known 
in practice. (See the discussion below Algorithm~\ref{alg3}.)

For the reader's convenience 
we put in boxes those steps of idealized 
TAD which are different from modified TAD.

\begin{algorithm}[Idealized TAD]\label{alg3}
Let $X_0^{lo}$ be in the basin $D$, 
set the low temperature clock time to zero: $T_{tad} =0$, 
let $T_{corr}>0$ be a (basin-dependent) decorrelation 
time, and iterate on the visited basins the following:

\begin{enumerate}[leftmargin=0.79in]
\item[]{\bf Decorrelation step:}
\item[1.]{Starting at time $t = T_{tad}$, evolve $X_t^{lo}$ at temperature $\beta = \beta^{lo}$ 
according to~\eqref{1} in the current basin $D$.}
\item[2.]{If $X_t^{lo}$ exits $D$ at a time 
$T_{tad} + \tau < T_{tad} + T_{corr}$, 
then set 
\begin{equation*}
{\hat S}(t) = S(D), \hskip10pt t \in [T_{tad},T_{tad}+\tau],
\end{equation*}
advance the low temperature clock by $\tau$:
$T_{tad} = T_{tad} + \tau$,
then go back to Step 1, where $D$ is now the 
new basin. Otherwise, set
\begin{equation*}
{\hat S}(t) = S(D), \hskip10pt t \in [T_{tad},T_{tad}+T_{corr}],
\end{equation*}
advance the low temperature clock by $T_{corr}$: $T_{tad} = T_{tad} + T_{corr}$, 
and initialize the exit step by setting $T_{sim} = 0$ and $T_{stop} = \infty$. 
Then proceed to the exit step.}
\end{enumerate}

\begin{enumerate}[leftmargin=0.79in]
\item[]{\bf Exit step:}
\end{enumerate}

\fcolorbox{black}[HTML]{E9F0E9}{\parbox{\textwidth}{
\begin{enumerate}
\item[1.]{Sample $X_{T_{sim}}^{hi}$ from the QSD at high temperature in $D$: 
$X_{T_{sim}}^{hi} \sim \nu^{hi}$.}
\end{enumerate}}}

\begin{enumerate}[leftmargin=0.79in]
\item[2.]{Evolve $X_t^{hi}$ at $\beta = \beta^{hi}$ starting 
at $t=T_{sim}$ until the 
first time after $T_{sim}$ at which it exits $D$.
Call this time $T_{sim}+\tau$.}
\item[3.]{Record the set $\partial D_i$ through which $X_t^{hi}$ exited $D$.}
\item[4.]{Advance the simulation clock by $\tau$: $T_{sim} = T_{sim} + \tau$.}
\end{enumerate}

\fcolorbox{black}[HTML]{E9F0E9}{\parbox{\textwidth}{
\begin{enumerate}
\item[5.]{If an exit through $\partial D_i$ has already been observed, go to Step 8. 
If an exit through $\partial D_i$ has not yet been observed, set $T_i^{hi} = T_{sim}$ and:
\begin{equation}\label{idealarrhenius}
T_i^{lo} = T_{i}^{hi}\,\Theta_i, \hskip20pt \Theta_i \equiv \frac{\lambda^{hi}p_i^{hi}}{\lambda^{lo}p_i^{lo}}.
\end{equation}
}
\end{enumerate}}}

\begin{enumerate}[leftmargin=0.79in]
\item[6.]{Update the lowest extrapolated exit time and corresponding exit spot:
\begin{align*}
T_{min}^{lo} &= \min\{T_{min}^{lo}, T_i^{lo}\}\\ 
I_{min}^{lo} &= i \hskip10pt\hbox{if}\hskip10pt T_{min}^{lo} = T_i^{lo}.
\end{align*}
}
\end{enumerate}

\fcolorbox{black}[HTML]{E9F0E9}{\parbox{\textwidth}{
\begin{enumerate}
\item[7.]{Update $T_{stop}$:
\begin{equation}\label{stop2}
T_{stop} = T_{min}^{lo}/C,\hskip20pt C \le \min_{1\le i\le k}\Theta_i. 
\end{equation}
}
\end{enumerate}}}
\begin{enumerate}[leftmargin=0.79in]
\item[8.]{If $T_{sim} \le T_{stop}$, go back to Step 1 of the exit step; 
otherwise, proceed to Step 9.}
\end{enumerate}

\begin{enumerate}[leftmargin=0.79in]
\item[9.]{Set 
\begin{equation*}
 {\hat S}(t) = S(D) \hskip10pt {for} \hskip10pt t \in [T_{tad},T_{tad}+T_{min}^{lo}],
\end{equation*}
and advance the low temperature simulation clock by 
$T_{min}^{lo}$:
\begin{equation*}
T_{tad} = T_{tad}+T_{min}^{lo}.
\end{equation*}}
\end{enumerate}

\fcolorbox{black}[HTML]{E9F0E9}{\parbox{\textwidth}{
\begin{enumerate}
\item[10.]{Let 
\begin{equation*}
X_{T_{tad}}^{lo} \sim \left[\rho^{lo}\left(\partial D_{I_{min}^{lo}}\right)\right]^{-1} \rho^{lo}|_{\partial D_{I_{min}^{lo}}}.
\end{equation*}
Then go back to the decorrelation step, the basin $D$ now 
being the one obtained by exiting through $X_{T_{tad}}^{lo}$.}
\end{enumerate}}}

\end{algorithm}
\vskip10pt

Below we comment in more detail on idealized TAD. 
\begin{itemize}
\item {\bf The quasistationary distribution in $D$: (I1) and (A1).}

In idealized TAD, the convergence to local equilibrium (see (M1)
and (M2) above) is assumed to
be reached, and this is made precise using 
the QSD $\nu$. In particular, we start the 
high temperature exit search exactly at the QSD $\nu^{hi}$; see (I1). 
We will also assume the low temperature dynamics reaches $\nu^{lo}$ at the 
end of the decorrelation step: 
\begin{itemize}
\item[(A1)] After the decorrelation step of idealized TAD, the low temperature dynamics is 
distributed according to the QSD in $D$: $X_{T_{tad}}^{lo} \sim \nu^{lo}$. 
\end{itemize}
This will be crucial for extrapolating the exit event 
at low temperature. Assumption (A1) is justified by the 
fact that the law of $X_t^{lo}$ in the decorrelation step 
approaches $\nu^{lo}$ exponentially fast in $T_{corr}$; 
see \cite{Tony, Gideon} for details. We also refer to~\cite{Tony,Gideon} 
for a presentation of algorithms which can be used to sample the QSD. 

\item{\bf The exit position: (I2).} 

To get exact metastable dynamics, 
we have to assume that the way the dynamics leaves $D$ near a 
given saddle point $x_i$ does not affect the 
metastable dynamics in the decorrelation step; see (I2). This 
can be justified in the small temperature regime by using Theorem~\ref{theorem0b} 
and some exponential decay results on the normal derivative of the QSD 
away from saddle points. Indeed, the conditional probability that, given the dynamics leaves 
through $\partial D_i$, it leaves outside a neighborhood of $x_i$ is of order $e^{-c\beta}$ 
as $\beta \to \infty$ (for a constant $c>0$); see \cite{Helffer,Tony2}. 

\item {\bf Replacing the Arrhenius law extrapolation rule: (I3).} 

In idealized TAD, we replace the extrapolation formula~\eqref{arrhenius} based
on the Arrhenius law by the idealized formulas~\eqref{ratios}-~\eqref{extrapolate}; 
see (I3). This is a
severe modification, since it makes the algorithm
impractical. In particular the quantities $\lambda^{lo}$ and
$p^{lo}_i$ are not known: if they were, it would be very easy to
simulate the exit event from~$D$; see Corollary~\ref{corollary1}
above. 

It is the aim of Section~\ref{sec:theta} below to explain how the small temperature
assumption is used to get practical estimates of the ratios $\Theta_i$. For simplicity we 
perform this small temperature analysis in one dimension. 
We will show that $\Theta_i$ is indeed close to the formula 
$\exp[-(\beta^{hi}-\beta^{lo})(V(x_i)-V(x_0))]$ used in the
original and modified TAD; compare~\eqref{idealarrhenius} 
with~\eqref{arrhenius} and~\eqref{arrhenius2}. 
We expect the same relation to be true in higher dimensions 
under appropriate conditions; this will be the subject of another paper.
\end{itemize}

In the analysis below, we need idealizations (I1) and (I3) to exactly 
replicate the law of the low temperature exit time and exit region in the exit 
step; see Theorem~\ref{theorem1} below. 
With (I1) and (I3), the inferred low temperature exit events are
statistically exact. This is based in particular on (A1), namely the
fact that the low temperature process is distributed according to
$\nu^{lo}$ at the end of the decorrelation step. In addition, after an
exit event, the
dynamics in the next decorrelation step depends on the exact 
exit {\it point} in $\partial D_i$: this is why we also need (I2) 
to get exact metastable dynamics; see Theorem~\ref{mainthm} below.

\subsection{Idealized TAD is exact}\label{section4}

The aim of this section is to prove the following result:
\begin{theorem}\label{mainthm}
Let $X_t^{lo}$ evolve according to~\eqref{1} at $\beta = \beta^{lo}$. 
Let ${\hat S(t)}$ be the metastable dynamics produced by 
Algorithm~\ref{alg3} (idealized TAD), assuming (A1), and let 
idealized TAD have the same initial condition as $X_t^{lo}$. Then: 
\begin{equation*}
{\hat S}(t)_{t\ge 0} \sim S(X_t^{lo})_{t\ge 0},
\end{equation*}
that is, the metastable dynamics produced by idealized TAD 
has the same law as the (exact) low temperature metastable dynamics.
\end{theorem}

Due to Corollary~\ref{corollary1}, (A1), (I2), and the fact that the 
low temperature dynamics is simulated exactly during the 
decorrelation step, it suffices to prove that the exit step of idealized 
TAD is exact in the following sense:

\begin{theorem}\label{theorem1}
Let $X_t^{lo}$ evolve according to~\eqref{1} at $\beta = \beta^{lo}$ with 
$X_t^{lo}$ initially distributed according to the QSD in $D$: 
$X_0^{lo} \sim \nu^{lo}$. Let $\tau = \inf \{t>0\,:\,X_t^{lo} \notin D\}$ and $I$ 
be the discrete random variable defined by: for $i=1,\ldots,k$,
 $$I=i \text{ if and only if } X_\tau^{lo} \in \partial D_i.$$ Let $T_{min}^{lo}$ and 
$I_{min}^{lo}$ be the random variables produced by the exit step of
idealized TAD. Then, $(T_{min}^{lo}, I_{min}^{lo})$ has the 
same probability law as $(\tau,I)$:
\begin{equation*}
(T_{min}^{lo}, I_{min}^{lo}) \sim (\tau,I).
\end{equation*}
\end{theorem}
The proof of Theorem~\ref{theorem1} will use (I1) and (I3) in particular. 
The theorem shows that the exit event from $D$ produced by idealized TAD
is exact in law compared to the exit event that would have
occurred at low temperature: the random variable $(T_{min}^{lo},I_{min}^{lo})$ 
associated with idealized TAD has the same law as the first exit time 
and location (from $D$) of a dynamics $(X_t^{lo})_{t \ge 0}$ obeying~\eqref{1} with
$\beta=\beta^{lo}$ and $X_0^{lo} \sim \nu^{lo}$.

To begin, we provide a simple lemma which shows that 
we can assume $T_{stop}\equiv \infty$ without loss 
of generality. We need this result in order to properly 
define all the random variables $T^{hi}_i$, for $i=1, \ldots,k$,
where we recall $k$ denotes the number of saddle points
of $V$ on
$\partial D$.
\begin{lemma}\label{lem:Tstop}
Consider the exit step of the idealized TAD, and modify Step 8 as follows:
\begin{itemize}
 \item[8.]{\it Go back to Step 1 of the exit step.}
\end{itemize}
Thus we loop between Step 1 and Step 8 of the exit step for infinite time, regardless of the values of 
$T_{sim}$ and $T_{stop}$. Then, $(T^{lo}_{min}, I_{min}^{lo})$ 
remains constant for all times $T_{sim} > T_{stop}$.
\end{lemma}
\begin{proof}
We want to show that without ever advancing to Step 10, the exit step of idealized TAD produces 
the same random variable $(T_{min}^{lo}, I_{min}^{lo})$ as soon as $T_{sim} > T_{stop}$. To 
see this, note that if $T_i^{lo} < T_{min}^{lo}$, 
then from~\eqref{idealarrhenius},
\begin{equation*}
 T_i^{lo} =  T_i^{hi}\frac{\lambda^{hi} p_i^{hi}}{\lambda^{lo}p_i^{lo}} < T_{min}^{lo}
\end{equation*}
and so, comparing with~\eqref{stop2}, 
\begin{equation*}
 T_i^{hi} < T_{min}^{lo}\frac{\lambda^{lo} p_i^{lo}}{\lambda^{hi} p_i^{hi}}\le 
\frac{T_{min}^{lo}}{C} = T_{stop}.
\end{equation*}
Thus, if $T_{sim} > T_{stop}$, any escape event will lead to an
extrapolated time $T_i^{lo}$ which will be larger than $T_{min}^{lo}$,
and thus will not change the value of $T_{min}^{lo}$ anymore.
\end{proof}

Let us now identify the laws of the random variables 
$(T_i^{hi})_{1\le i \le l}$ produced by idealized TAD.
\begin{proposition}\label{prop:Thi}
Consider idealized TAD in the setting of Lemma~\ref{lem:Tstop}, 
so that all the $T^{hi}_i$ are defined, $i=1,2,\ldots,k$.

Let $(\tau^{(j)},I^{(j)})_{j \ge 1}$ be independent and
identically distributed random 
variables such that $\tau^{(j)}$ is independent from 
$I^{(j)}$, $\tau^{(j)} \sim {\cal E}(\lambda^{hi})$ and 
for $i=1,\ldots,k$, $I^{(j)}$ is 
a discrete random variable with law 
\begin{equation*}
 \PP(I^{(j)}=i) = p^{hi}_i.
\end{equation*}
For $i=1,\ldots,k$ define 
\begin{equation}\label{defineNT1}
N_i^{hi} = \min\{j\,:\,I^{(j)}=i\}.
\end{equation}
Then we have the following equality in law:
\begin{equation}\label{eq:T}
(T^{hi}_1, \ldots, T^{hi}_k)\sim
\left(\sum_{j=1}^{N_1^{hi}}\tau^{(j)}, \ldots,
  \sum_{j=1}^{N_k^{hi}}\tau^{(j)}\right).
\end{equation}
Moreover, (i) $T^{hi}_i \sim {\cal E}(\lambda^{hi} p^{hi}_i)$ and (ii) $T^{hi}_1,T^{hi}_2,\ldots,T^{hi}_k$ are independent.
\end{proposition}
\begin{proof}
The equality~\eqref{eq:T} follows from Corollary~\ref{corollary1}, 
since in the exit step of idealized TAD, the dynamics restarts from 
the QSD $\nu^{hi}$ after each escape event. 

Let us now consider the statement $(i)$.  Observe that the moment generating function of an exponential 
random variable $\tau$ with parameter $\lambda$ is: for $s < \lambda$,
\begin{equation*}
 \E\left[\exp\left(s\tau\right)\right] = \int_{0}^\infty e^{st}\lambda e^{-\lambda t}\,dt = \frac{\lambda}{\lambda-s}.
\end{equation*}
So, dropping the superscript $hi$ for ease of notation, we have:
for $i \in \{1,\ldots,k\}$, and for $s < \lambda p_i$,
\begin{align*}
\E\left[\exp\left(s T_i\right)\right] &= \sum_{m=1}^\infty\E\left[\exp\left(s T_i\right)\Big| N_i = m\right]\PP\left(N_i = m\right) \\
&= \sum_{m=1}^\infty \E\left[\exp\left(s \sum_{j=1}^{m}\tau^{(j)}\right)\right](1-p_i)^{m-1}p_i\\
&= \sum_{m=1}^\infty \E\left[\exp\left(s \tau^{(1)}\right)\right]^m(1-p_i)^{m-1}p_i\\
&= \frac{\lambda p_i}{\lambda - s}\sum_{m=1}^\infty \left(\frac{\lambda\left(1 -  p_i\right)}{\lambda - s}\right)^{m-1}\\
&= \frac{\lambda p_i}{\lambda p_i - s}.
\end{align*}
This shows $T^{hi}_i \sim {\cal E}(\lambda^{hi} p^{hi}_i)$.
\end{proof}

Before turning to the proof of the statement $(ii)$ in Proposition~\ref{prop:Thi}, we need the following 
technical lemma:
\begin{lemma}\label{lemmasym}
Let $a_1,a_2,\ldots,a_n$ be positive real numbers, and let $S_n$ be the symmetric 
group on $\{1,2,\ldots,n\}$. Then 
\begin{equation}\label{symmetric}
\sum_{\sigma \in S_n}\prod_{i=1}^n \left(\sum_{j=i}^n a_{\sigma(j)}\right)^{-1} = \prod_{i=1}^n a_i^{-1}.
\end{equation}
\end{lemma}
\begin{proof} Note that~\eqref{symmetric} is of course true for $n=1$. Assume it 
is true for $n-1$, and let 
\begin{equation*}
S_n^{(k)} = \{\sigma \in S_n\,:\, \sigma(1) = k\}.
\end{equation*}
Then
\begin{align*}
\sum_{\sigma \in S_n} \prod_{i=1}^n \left(\sum_{j=i}^n a_{\sigma(j)}\right)^{-1} &= 
\left(\sum_{i=1}^n a_i\right)^{-1}\sum_{\sigma \in S_n}\prod_{i=2}^n \left(\sum_{j=i}^n a_{\sigma(j)}\right)^{-1} \\
&=\left(\sum_{i=1}^n a_i\right)^{-1}\sum_{k=1}^n \sum_{\sigma \in S_n^{(k)}}\prod_{i=2}^n\left(\sum_{j=i}^n a_{\sigma(j)}\right)^{-1}\\
&=\left(\sum_{i=1}^n a_i\right)^{-1}\sum_{k=1}^n \prod_{\substack{j=1\\j\ne k}}^n a_j^{-1} \\
&= \prod_{i=1}^n a_i^{-1}.
\end{align*}
By induction~\eqref{symmetric} is valid for all $n$.
\end{proof}

We are now in position to prove statement $(ii)$ of
Proposition~\ref{prop:Thi}.
\begin{proof}[Proof of Proposition~\ref{prop:Thi} part $(ii)$]
In this proof, we drop the superscript $hi$ for ease of notation. To show that the $T_i$'s are independent, it suffices to show that for 
$s_1,\ldots,s_k$ in a neighborhood of zero we have 
\begin{equation}\label{toshow}
\E\left[\exp\left(\sum_{i=1}^k s_i T_i\right)\right] = \prod_{i=1}^k \E\left[\exp\left(s_i T_i\right)\right].
\end{equation}
We saw in the proof of part $(i)$ that: for $s_i < \lambda p_i$,
\begin{equation}\label{factors}
 \E\left[\exp\left(s_i T_i\right)\right] = \frac{\lambda p_i}{\lambda p_i -s_i}.
\end{equation}
Consider then the left-hand-side of~\eqref{toshow}. We start by a
preliminary computation. Let $m_0 = 0$, $m_1 = 1$, 
and $s_i < \lambda p_i$ for $i=1,\ldots,k$. Then 
\begin{align}
&\sum_{1<m_2<m_3\ldots<m_k} \E\left[\exp\left(\sum_{i=1}^k s_i T_i\right)\Big| \cap_{i=1}^k \{N_i = m_i\}\right]\PP\left(\cap_{i=1}^k \{N_i = m_i\}\right) \nonumber \\
&=\sum_{1<m_2<m_3\ldots<m_k} \E\left[\exp\left(\sum_{i=1}^k \left(s_i \sum_{j=1}^{m_i}\tau^{(j)}\right)\right)\right]
p_1 \prod_{i=2}^k p_i\left(1-\sum_{j=i}^k p_j\right)^{m_i-m_{i-1}-1}\nonumber\\
&= p_1\sum_{1<m_2<m_3\ldots<m_k}\,\prod_{i=1}^k  \E\left[\exp\left(\left(\sum_{j=i}^k s_j\right)\sum_{j=m_{i-1}+1}^{m_i}\tau^{(j)}\right)\right]\prod_{i=2}^k p_i\left(1-\sum_{j=i}^k p_j\right)^{m_i-m_{i-1}-1} \nonumber\\
&= p_1 \sum_{1<m_2<m_3\ldots<m_k}\,\prod_{i=1}^k \E\left[\exp\left(\tau^{(1)}\sum_{j=i}^k s_j\right)\right]^{m_i-m_{i-1}}\prod_{i=2}^k p_i\left(1-\sum_{j=i}^k p_j\right)^{m_i-m_{i-1}-1} \nonumber\\
&=\left(\frac{\lambda p_1}{\lambda - \sum_{j=1}^k s_j}\right) \sum_{1<m_2<m_3\ldots<m_k}\,\prod_{i=2}^k p_i \left(\frac{\lambda}{\lambda - \sum_{j=i}^k s_j}\right)\left(\frac{\lambda\left(1-\sum_{j=i}^k p_j\right)}{\lambda - \sum_{j=i}^k s_j}\right)^{m_i-m_{i-1}-1}\nonumber\\
&= \left(\frac{\lambda p_1}{\lambda - \sum_{j=1}^k s_j}\right)\prod_{i=2}^k p_i \left(\frac{\lambda}{\lambda - \sum_{j=i}^k s_j}\right) 
\left(1-\frac{\lambda\left(1 - \sum_{j=i}^k p_j\right)}{\lambda - \sum_{j=i}^k s_j}\right)^{-1} \nonumber\\
&= \left(\frac{\lambda p_1}{\lambda - \sum_{j=1}^k s_j}\right)\prod_{i=2}^k \lambda p_i \left(\sum_{j=i}^k \lambda p_j - s_j\right)^{-1}\nonumber\\
&= \prod_{i=1}^k \lambda p_i \left( \sum_{j=i}^k \lambda p_j - s_j\right)^{-1}. \nonumber\\
\label{long}
\end{align}
From~\eqref{long} observe that
\begin{align}\begin{split}\label{assume}
\E\left[\exp\left(\sum_{i=1}^k s_i T_i\right)\right] 
&= \sum_{\sigma \in S_k}\prod_{i=1}^k \lambda p_{\sigma(i)} \left( \sum_{j=i}^k \lambda p_{\sigma(j)} - s_{\sigma(j)}\right)^{-1}\\
&= {\left(\prod_{i=1}^k \lambda p_i\right)} \sum_{\sigma \in S_k}\prod_{i=1}^k \left( \sum_{j=i}^k \lambda p_{\sigma(j)} - s_{\sigma(j)}\right)^{-1}\\
&= \prod_{i=1}^k \frac{\lambda p_i}{\lambda p_{i} - s_{i}},
\end{split}
\end{align}
where in the last step we have used Lemma~\ref{lemmasym}. Comparing 
~\eqref{toshow} with~\eqref{factors} and~\eqref{assume}, we are done.
\end{proof}

To complete the proof of Theorem~\ref{theorem1}, we finally need the
following Lemma.
\begin{lemma}\label{lemma2}
Let $T_1,\ldots,T_k$ be independent random variables such that $T_i
\sim {\cal E}(\lambda p_i)$, with $\lambda >0$, $p_i > 0$ and
$\sum_{j=1}^k p_i =1$. Set 
\begin{equation*}
T = \min_i T_i \hskip10pt\hbox{ and }\hskip10pt I = \arg\min_i\, T_i.
\end{equation*}
Then: (i) $T \sim {\cal E}(\lambda)$, (ii) $\PP(I = i) = p_i$, and (iii) $T$ and $I$ are independent. 
\end{lemma}
\begin{proof}
Since the $T_i$'s are assumed to be independent, it is well known that $T = T_I = \min_i T_i$ is an exponential random variable with 
parameter $\sum_i \lambda p_i = \lambda$. This proves $(i)$. 
Turning to $(ii)$ and $(iii)$, note that $\min_{j\ne i}T_j$ is an exponential 
random variable independent of $T_i$ with parameter 
\begin{equation*}
\sum_{j\ne i} \lambda p_j = \lambda (1-p_i).
\end{equation*}
Thus, 
\begin{align}\begin{split}\label{iiandiii}
 \PP(I = i,  T_{I} \ge t) &= \PP(t \le T_i \le \min_{j\ne i}T_j)\\
&=\int_t^\infty \int_s^\infty \lambda p_i e^{-\lambda p_i s} \, \lambda(1-p_i) e^{-\lambda (1-p_i)r}
\,dr\,ds  \\
&= \int_t^\infty \lambda p_i e^{-\lambda s}\,ds \\
&= p_i \PP(T_I \ge t).
\end{split}
\end{align}
Setting $t=0$ we obtain $\PP(I = i) = p_i$, which proves $(ii)$. Now 
$(iii)$ follows from~\eqref{iiandiii}.
\end{proof}

We are now in position to prove Theorem~\ref{theorem1}.
\begin{proof}[Proof of Theorem~\ref{theorem1}.] 
First, by Lemma~\ref{lem:Tstop}, we can assume that $T_{stop} =
\infty$ so that all the $T_i^{hi}$'s are well defined, for $i=1,
\ldots, k$. Then Proposition~\ref{prop:Thi} implies that the
$T_i^{hi}$'s are independent exponential random variables with parameters 
$\lambda^{hi} p_i^{hi}$. So by~\eqref{idealarrhenius}, the $T_i^{lo}$'s are 
independent exponential random variables with parameters $\lambda^{lo} p_i^{lo}$. 
Now by applying Lemma~\ref{lemma2} to the $T_i^{lo}$'s, we get $T_{min}^{lo} \sim {\cal E}(\lambda^{lo})$, 
$\PP(I_{min}^{lo} = i) = p_i^{lo}$, and $T_{min}^{lo}$ is independent of $I_{min}^{lo}$. Referring to Corollary~\ref{corollary1}, we are done. 
\end{proof}

\begin{remark}
Observe that the proof of Theorem~\ref{theorem1} does not use (I2), 
which is needed only to obtain correct metastable dynamics by iterating the exit step. 
Also, notice that we did not use the fact that $D$ is the basin of 
attraction of a local minimum of $V$ or that each set $\partial D_i$
in the partition of $\partial D$ is associated to a saddle point $x_i$ for the moment. 
The latter assumption is crucial in the next section, in which we obtain 
computable estimates of the ratios $\Theta_i$, $i=1,\ldots,k$;
this will also require an assumption of large $\beta$ which was not needed 
for Theorem~\ref{theorem1}.
\end{remark}

\section{Estimates for the $\Theta_i$'s at low temperature in one dimension}\label{sec:theta}

In the last section we showed that modified TAD (Algorithm~\ref{alg2}) 
is {\it exact} with the idealizations (I1)-(I3) and the assumption (A1); 
see idealized TAD (Algorithm~\ref{alg3}). In this section we justify (I3). 
In particular, we show in Theorem~\ref{theorem2} below how the ratios 
$\Theta_i$ (see~\eqref{ratios}) can be approximated by explicit practical 
formulas in one dimension. Compared to Theorem~\ref{theorem1}, the proof of 
Theorem~\ref{theorem2} will require the additional assumption that temperature 
is sufficiently small.

\subsection{Statement of the main result}
We recall that the ratios $\Theta_i$, $i=1,\ldots,k$ are unknown in
practice. In TAD these ratios are approximated using
the Arrhenius law. The main result of this section, Theorem~\ref{theorem2}, gives precise 
asymptotics for $\Theta_i$ as $\beta^{hi}, \beta^{lo} \to \infty$. 
In particular, we show that $\Theta_i$ converges to 
$\exp[-(\beta^{hi}-\beta^{lo})(V(x_i)-V(x_0))]$.

Throughout this section we assume that we are {\em in a one dimensional
setting}. Moreover, we assume that $D$ is the basin of attraction of the gradient
dynamics $dy/dt = -V'(y)$ associated to a local minimum
of $V$ (this is what is done in practice by
A.F. Voter and co-workers). Finally, the potential $V$ is assumed to be
a Morse function, which means that the critical points of $V$ are non-degenerate.
Under these assumptions, we may assume
without additional loss of generality that (see Figure~\ref{fig2}):
\begin{itemize}
 \item[(B1)]{$D = (0,b)$, with $b>1$, $V(0)=0$, and $V'(x) \ne 0$ for $x \notin \{0,1,b\}$,}

\item[(B2)]{$V'(0) = 0 = V'(b)$ and $V''(0)<0$, $V''(b)<0$,}

\item[(B3)]{$V'(1) = 0$ and $V''(1)>0$.}
\end{itemize}
We also normalize $u$ (see~\eqref{eq:ulambda}) so that 
\begin{itemize}
\item[(B4)] $u(1) = 1$.
\end{itemize}
In particular, the location of the minimum of $V$ and the value 
of $V$ at $0$ are chosen for notational convenience and without loss of 
generality. In the following, we write $\{0\} = \partial D_1$ and $\{b\} = \partial D_2$. 

\begin{figure}
\begin{center}
\includegraphics[scale=0.5]{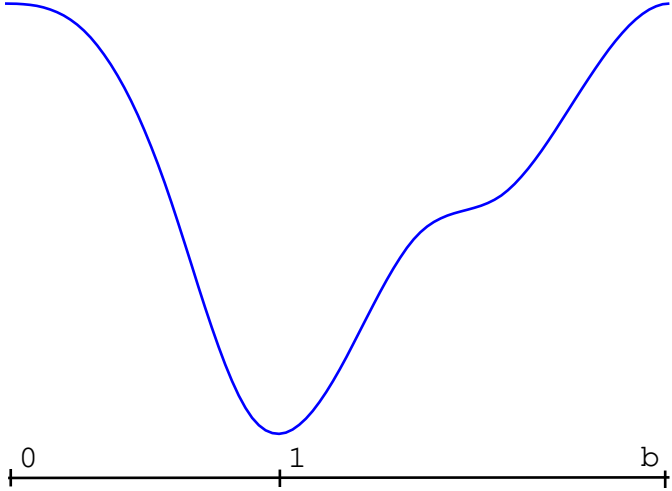}
\end{center}
\caption{A function $V:D \to {\mathbb R}$ satisfying
  (B1)-(B3).}
\label{fig2}
\end{figure}

We will prove the following:
\begin{theorem}\label{theorem2}
Under the assumptions stated above, we
have the formula: for $i=1,2$,
\begin{equation}\label{main}
\Theta_i=\frac{\lambda^{hi}p_i^{hi}}{\lambda^{lo}p_i^{lo}} = 
e^{-(\beta^{hi}-\beta^{lo})(V(x_i)-V(x_0))}\left(1 + O\left(\frac{1}{\beta^{hi}}- \frac{1}{\beta^{lo}}\right)\right)
\end{equation}
as $\beta^{hi},\beta^{lo} \to \infty$, $\beta^{lo}/\beta^{hi}=r$ 
where $x_1 = 0$, $x_2 = b$ and $x_0 = 1$, and $r>0$ is constant.
\end{theorem}

The ratios $\frac{\lambda^{hi}p_i^{hi}}{\lambda^{lo}p_i^{lo}}$ involve integrals of the form $\int_D e^{-\beta V(x)} u(x)\,dx$ 
at high and low temperature. We will use Laplace expansions to 
analyze the integrals, but since $u$ depends on $\beta$, extra care must 
be taken in the analysis. 

\subsection{Proof of Theorem~\ref{theorem2}}

In all what follows, $(u,-\lambda)$ denotes the principal 
eigenvector/eigenvalue pair of $L$ with homogeneous
Dirichlet boundary conditions; see~\eqref{eq:ulambda}. We are
interested in how the pair $(u,-\lambda)$ varies in the small
temperature regime $\beta \to \infty$.

Throughout this section, we write $c$ to denote a {\it positive} constant, 
the value of which may change without being explicitly noted.
To begin, we will need some asymptotics for $\lambda$ and $u$, Lemma~\ref{lemma4} and 
Lemma~\ref{lemma5} below. The contents of both lemmas are found in or 
implied by \cite{Day2}, 
\cite{Devinatz}, and \cite{Friedman} (see also \cite{Ofer} and \cite{Day}) in the case where 
$V' \cdot n >0$ on $\partial D$, with $n$ the normal to $\partial D$
(in our setting $n=1$ on $\partial D_2$ and $n=-1$ on $\partial D_1$). Here, we 
consider the case of {\it characteristic boundary}, where 
from (B2) $V' \cdot n = 0$ on $\partial D$, so we adapt 
the classical results to this case.

\begin{lemma}\label{lemma4}
There exists $c>0$ such that 
\begin{equation}\label{lambda}
\lambda = O\left(e^{-c\beta}\right)\hskip20pt\hbox{as }\beta \to \infty.
\end{equation}
\end{lemma}
\begin{proof}
Let $D' \subset D$ be a domain containing $1$ such that ${\overline {D'}} \subset D$, and 
let $(u',-\lambda')$ the principal eigenvector/eigenvalue pair 
for $L$ on $D'$ with homogeneous Dirichlet boundary conditions on $\partial D'$. 
Recall that $\lambda$ is given by the Rayleigh formula
\begin{equation*}
 \lambda = \inf_{f \in H_V^1(D)}\frac{\beta^{-1}\int_D |\nabla f(x)|^2\, e^{-\beta V(x)}\,dx}
{\int_D f(x)^2 \,e^{-\beta V(x)}\,dx},
\end{equation*}
where $H_V^1(D)$ is the space of functions vanishing on ${\mathbb R}\setminus D$ such that 
\begin{equation*}
\int_D \left(|\nabla f(x)|^2 + f(x)^2\right)e^{-\beta V(x)}\,dx < \infty,
\end{equation*} 
and similarly for $\lambda'$. 
Since every function vanishing on ${\mathbb R}\setminus D'$ also vanishes 
on ${\mathbb R} \setminus D$, we have 
\begin{equation}\label{lambdaprime}
 \lambda \le \lambda'.
\end{equation}
Now let $X_t^1$ obey~\eqref{1} with $X_0^1 = 1$, and define $\tau' = \inf\{t>0\,:\, X_t^1 \notin D'\}$. 
Since $D'$ is a sub-basin of attraction such that $V'$ points outward on $\partial D'$, 
we can use the following classical results (see e.g. Lemmas 3--4 of~\cite{Day2}):
\begin{equation}\label{taulambda}
\lim_{\beta \to \infty} \beta^{-1}\log \E[1/\lambda']=\lim_{\beta \to \infty} \beta^{-1} \log \E[\tau'] = \inf_{z \in \partial D'}\inf_{t>0}\, I_{z,t}
\end{equation}
where, by definition,
\begin{align*}
&I_{z,t} = \inf_{f \in H_1^z[0,t]} \frac{1}{4}\int_0^t |{\dot f}(s)+V'(f(s))|^2\,ds\\
&H_1^z[0,t] = \left\{f\,:\, \exists {\dot f} \in L^2[0,t]\,\,s.t.\,\,f(t) = z,\,\forall s \in [0,t],
\,f(s) = 1 + \int_0^s {\dot f}(r)\,dr\right\}.
\end{align*}
Observe that for any $t>0$ and $f \in H_1^z[0,t]$ we have 
\begin{align*}
&\frac{1}{4}\int_0^t \left|{\dot f}(s) + V'(f(s))\right|^2\,ds \\
&= \frac{1}{4} \int_0^t \left|{\dot f}(s) - V'(f(s))\right|^2\,ds + \int_0^t {\dot f}(s) V'(f(s))\,ds\\
&\ge V(z)-V(1).
\end{align*}
Since $\partial D'$ is disjoint from $1$ we can conclude that for $z \in \partial D'$,  
$I_{z,t} \ge c > 0$ uniformly in $t>0$, for a positive constant $c$. Thus, 
\begin{equation*}
\lim_{\beta \to \infty} \beta^{-1} \log \E[\tau'] \ge c > 0
\end{equation*}
which, combined with~\eqref{lambdaprime} and~\eqref{taulambda}, implies the result.
\end{proof}

Next we need the following regularity result for $u$:
\begin{lemma}\label{lemma5}
The function $u$ is uniformly bounded in $\beta$, that is, 
\begin{equation}
||u||_{\infty} = O(1)\hskip20pt \hbox{as } \beta \to \infty,
\end{equation}
where $||\cdot||_{\infty}$ is the $L^\infty$ norm on $C[0,b]$. 
\end{lemma}

\begin{proof}
Define $f(t,x) = u(x)e^{\lambda t}$ and set 
\begin{equation*}
\tau^x = \inf\{t > 0\,:\, X_t^x \notin (0,1)\}
\end{equation*}
where $X_t^x$ obeys~\eqref{1} with $X_0^x = x$. 
Fix $T>0$. By It\={o}'s lemma, for $t \in [0,T\wedge\tau^x]$ we have
\begin{align*}
f(t,X_t^x) &= u(x) + \lambda \int_0^t u(X_s^x)e^{\lambda s}\,ds + \int_0^t Lu(X_s^x)e^{\lambda s}\,ds 
+ \sqrt{2\beta^{-1}}\int_0^t u'(X_s^x)\,dW_s \\
&= u(x) + \sqrt{2\beta^{-1}}\int_0^t u'(X_s^x)\,dW_s.
\end{align*}
Setting $t = T\wedge\tau^x$ and taking expectations gives 
\begin{equation}\label{FK}
u(x) = \E\left[f(T\wedge\tau^x,X_{T\wedge\tau^x}^x)\right] = {\E}\left[e^{\lambda T\wedge\tau^x}u(X_{T\wedge\tau^x}^x)\right].
\end{equation}
Recall that $u$ is bounded for fixed $\beta$. We show in 
equations~\eqref{bound1} below that 
$\E[e^{\lambda \tau^x}]$ is finite, so we may let $T \to \infty$ 
in~\eqref{FK} and use the dominated convergence theorem to obtain
\begin{equation}\label{ubound}
u(x) = {\E}\left[e^{\lambda \tau^x}u(X_{\tau^x}^x)\right] \le {\E}\left[e^{\lambda \tau^x}\right],
\end{equation}
where we have recalled $u(0)=0$ and, from (B4), $u(1) = 1$. 
The idea is then to compare $\tau^x$ to the first hitting time of $1$ 
of a Brownian motion reflected at zero. Define 
\begin{equation*}
\sigma^x = \inf\{t>0\,:\, B_t^x \notin (-1,1)\}
\end{equation*}
where 
\begin{equation*}
B_t^x = x + \sqrt{2\beta^{-1}}{W}_t
\end{equation*}
with ${W}_t^x$ as in~\eqref{1}. 
Let ${\bar B}_t^x$ and ${\bar X}_t^x$ be given by reflecting $B_t^x$ 
and $X_t^x$ at zero. Since $V' < 0$ on $(0,1)$, it is clear that 
${\bar X}_t^x \ge {\bar B}_t^x$ for each $x \in (0,1)$ and $t \ge 0$. 
Thus, 
\begin{align}\begin{split}\label{compare}
\PP\left(\tau^x \ge t\right)  
&\le \PP\left(\inf\{s>0\,:\, {\bar X}_s^x = 1\} \ge t\right)\\
&\le \PP\left(\inf\{s>0\,:\, {\bar B}_s^x = 1\} \ge t\right)\\
&\le \PP\left(\inf\{s>0\,:\, {\bar B}_s^0 = 1\} \ge t\right)\\
&= \PP(\sigma^0 \ge t).
\end{split}
\end{align}
We will bound from above the last line of~\eqref{compare}. Let
$v(t,x)$ solve the heat equation $v_t = \beta^{-1}v_{xx}$ with $v(0,x)=1$ 
for $x \in (-1,1)$ and $v(t,\pm 1) = 0$. An elementary analysis shows that
\begin{equation}\label{v1}
v(t,0) \le \frac{4}{\pi}\exp(-\beta^{-1} \pi^2 t/4).
\end{equation}
(The Fourier sine series for $v(t,x-1)$ on $[0,2]$ at $x=1$ is an alternating 
series, and its first term gives the upper bound above.)
We claim that for fixed $t$ and $x \in [0,t]$, 
\begin{equation}\label{v2}
 v(t,0) = \PP(\sigma^0 \ge t).
\end{equation}
To see this, let $w(s,x) = v(t-s,x)$ and observe that $w_s = - \beta^{-1}w_{xx}$, so by It\={o}'s lemma, for $s \in [0,t \wedge \sigma^x]$ 
\begin{align*}
 w(s,B_s^x) &= w(0,x) + \int_0^s \left(w_s + \beta^{-1}w_{xx}\right)(r,B_r^x)\,dr + \sqrt{2\beta^{-1}}\int_0^s w_x(r,B_x^r)\,dW_r\\
&=w(0,x)+ \sqrt{2\beta^{-1}}\int_0^s w_x(r,B_x^r)\,dW_r.
\end{align*}
By taking expectations and setting $s = t \wedge \sigma^x$ we obtain
\begin{align*}
v(t,x) =w(0,x) &= \E\left[w\left(t \wedge \sigma^x,B_{t \wedge \sigma^x}^x\right)\right] \\
&= \E\left[w\left(t,B_t^x\right)\,1_{\{t\le \sigma^x\}}\right] + \E\left[w\left(\sigma^x,B_{\sigma^x}^x\right)\,1_{\{t>\sigma^x\}}\right] \\
&= \E\left[v\left(0,B_{t}^x\right)\,1_{\{t\le \sigma^x\}}\right] \\
&= \PP(\sigma^x \ge t).
\end{align*}
From~\eqref{compare},~\eqref{v1} and~\eqref{v2}, for $x \in [0,1)$ 
\begin{equation*}
\PP(\tau^x \ge t) \le \frac{4}{\pi}\exp({-\beta^{-1} \pi^2 t/4}).
\end{equation*}
By Lemma~\ref{lemma4}, $\lambda \beta \to 0$ as $\beta \to \infty$. 
So for all sufficiently large $\beta$, 
\begin{align}\begin{split}\label{bound1}
\E\left[e^{\lambda \tau^x}\right] 
&= 1 + \int_1^\infty \PP(e^{\lambda \tau^x} \ge t)\,dt \\
&\le 1 + \frac{4}{\pi}\int_1^\infty  t^{-\pi^2/(4\lambda\beta)}\,dt\\
&= 1 + \frac{4}{\pi}\frac{4\lambda\beta}{\pi^2-4\lambda\beta}.
\end{split}
\end{align}
Now recalling~\eqref{ubound}, 
\begin{equation}\label{ubound2}
u(x) \le \E\left[e^{\lambda \tau^x}\right]\le 1 + \frac{4}{\pi}\frac{4\lambda \beta}{\pi^2 - 4 \lambda \beta}.
\end{equation}
Using Lemma~\ref{lemma4} we see that the right hand side of 
~\eqref{ubound2} approaches $1$ as $\beta \to \infty$. An 
analogous argument can be made for $x \in (1,b]$, showing that $u$ is 
uniformly bounded in $\beta$ as desired.
\end{proof}

Next we define a function which will be useful in the analysis of~\eqref{ratios}. 
For $x \in [0,1]$ let 
\begin{equation}\label{f}
 f(x) = \frac{\int_0^x e^{\beta V(t)}\,dt}{\int_0^1 e^{\beta V(t)}\,dt}.
\end{equation}
We compare $u$ and $f$ in the following lemma:

\begin{lemma}\label{lemma6}
Let $||\cdot||_{\infty}$ the $L^\infty$ norm on $C[0,1]$. 
With $f$ defined by~\eqref{f}, we have, in the limit $\beta \to \infty$,
\begin{align*}
&||f-u||_{\infty} = O\left(e^{-c\beta}\right),\\
&||f'-u'||_{\infty} = O\left(e^{-c\beta}\right).
\end{align*}
\end{lemma}
\begin{proof}
Observe that $g = f-u$, defined on $[0,1]$, satisfies 
\begin{align}\begin{split}\label{g}
&-V'(x)g'(x) + \beta^{-1}g''(x) = \lambda u(x) \\
&\hskip74pt g(0)=0, \hskip5pt g(1) = 0.
\end{split}
\end{align}
Multiplying by $\beta e^{-\beta V(x)}$ in~\eqref{g} 
leads to 
\begin{equation*}
\frac{d}{dx}\left(e^{-\beta V(x)}g'(x)\right) = \beta e^{-\beta V(x)} \lambda u(x)
\end{equation*}
so that 
\begin{equation}\label{gprime}
 g'(x) = e^{\beta V(x)}\left(\lambda \beta \int_0^x e^{-\beta V(t)}u(t)\,dt + C_\beta\right).
\end{equation}
Integrating~\eqref{gprime} and using $g(0)=0$, 
\begin{equation}\label{here}
g(x) = \lambda \beta \int_0^x \left(e^{\beta V(t)}\int_0^t  e^{-\beta V(s)}u(s)\,ds\right)\,dt + C_\beta \int_0^x e^{\beta V(t)}\,dt.
\end{equation}
Using Lemma~\ref{lemma5} we have $||u||_{\infty}\le K<\infty$.  
From (B1) and (B3) 
we see that $V$ is decreasing on $[0,1]$. 
So putting $g(1)=0$ in~\eqref{here} we obtain, 
for all sufficiently large $\beta$, 
\begin{align}\begin{split}
\label{use}
|C_\beta| &= \lambda \beta \left(\int_0^1 e^{\beta V(t)}\,dt\right)^{-1} \int_0^1 \left(e^{\beta V(t)}\int_0^t e^{-\beta V(s)}u(s)\,ds\right)\,dt \\
&\le \lambda \beta \left(\int_0^1 e^{\beta V(t)}\,dt\right)^{-1}\int_0^1 \left(\int_0^t u(s)\,ds\right)\,dt \\
 &\le \lambda \beta K \left(\int_0^1 e^{\beta V(t)}\,dt\right)^{-1} \\
 &\le 2\lambda \beta^{3/2} K \left(\frac{-2 V''(0)}{\pi}\right)^{1/2}
 \end{split}
\end{align}  
where in the last line Laplace's method is used. 
Using Lemma~\ref{lemma4}, for all sufficiently large $\beta$, 
\begin{equation}\label{bound}
|C_\beta| \le e^{-c\beta}.
\end{equation}
From (B1) and (B3) we see that $V$ is nonpositive on $[0,1]$, so 
from~\eqref{gprime}, 
\begin{align}\begin{split}\label{gprimebound}
|g'(x)|&\le \lambda \beta \int_0^x e^{\beta (V(x)-V(t))}u(t)\,dt + |C_\beta| e^{\beta V(x)} \\
&\le \lambda \beta K  + e^{-c\beta}.
\end{split}
\end{align}
Using Lemma~\ref{lemma4} again,
we get $||g'||_{\infty} = O(e^{-c\beta})$. 
As $g(0) = 0$ this implies $||g||_{\infty} = O(e^{-c\beta})$. 
This completes the proof.
\end{proof}
\begin{remark}\label{remark2}
A result analogous to Lemma~\ref{lemma5} holds, with 
\begin{equation*}
 f(x) = \frac{\int_x^b e^{\beta V(t)}\,dt}{\int_1^b e^{\beta V(t)}\,dt},
\end{equation*}
for $x \in [1,b]$. 
\end{remark}

We are now in position to prove Theorem~\ref{theorem2}.
\begin{proof}[Proof of Theorem~\ref{theorem2}]
It suffices to prove the case $i=1$, so we will look at the 
endpoint $\partial D_1 = \{0\}$. From Theorem~\ref{theorem0b} we have 
\begin{equation*}
\rho(\{0\}) = \frac{\frac{d}{dx}\left(u(x) e^{-\beta V(x)}\right)\big|_{x=0}}{\beta \lambda\int_D u(x)e^{-\beta V(x)}\,dx} 
\end{equation*}
so that 
\begin{equation*}
\lambda p_1 =   \frac{e^{-\beta V(0)} u'(0)}
{\beta\int_D u(x)e^{-\beta V(x)}\,dx}.
\end{equation*}
Introducing again the superscripts $^{hi}$ and $^{lo}$, 
\begin{equation}\label{C}
\frac{\lambda^{hi}p_1^{hi}}{\lambda^{lo}p_1^{lo}} = e^{-(\beta^{hi}-\beta^{lo})V(0)}\cdot \frac{\beta^{lo}}{\beta^{hi}}\cdot\frac{u^{hi\,'}(0)}
{u^{lo\,'}(0)}
\cdot\frac{\int_D u^{lo}(x)e^{-\beta^{lo} V(x)}\,dx}{\int_D u^{hi}(x)e^{-\beta^{hi} V(x)}\,dx} 
\end{equation}
Dropping the superscripts, recalling the function $f$ from~\eqref{f}, 
and using Lemma~\ref{lemma6}, we see that 
\begin{equation*}
u'(0) = f'(0) + O\left(e^{-c\beta}\right).
\end{equation*}
Since
\begin{equation*}
 f'(0) = \left(\int_0^1 e^{\beta V(t)}\,dt\right)^{-1} =\left(1+k_1\beta^{-1}+O\left(\beta^{-2}\right)\right)\sqrt{\beta}\left(\frac{-2V''(0)}{\pi}\right)^{1/2}
\end{equation*}
where $k_1$ is a $\beta$-independent constant coming from the 
second term in the Laplace expansion. Thus 
\begin{align}\begin{split}\label{uprime}
u'(0) &= O\left(e^{-c\beta}\right)+\left(1+k_1\beta^{-1}+O\left(\beta^{-2}\right)\right)\sqrt{\beta}\left(\frac{-2V''(0)}{\pi}\right)^{1/2}\\
&= \left(1+k_1\beta^{-1}+O\left(\beta^{-2}\right)\right)\sqrt{\beta}\left(\frac{-2V''(0)}{\pi}\right)^{1/2}.
\end{split}
\end{align}
This takes care of the third term of the product in~\eqref{C}. 
We now turn to the fourth term. Let $y \in (0,1]$ and note that for $t \in (y,1]$, 
\begin{equation*}
f'(t) = e^{\beta V(t)}\left(\int_0^1 e^{\beta V(x)}\,dx\right)^{-1}  = O\left(e^{-c\beta}\right)
\end{equation*}
where here $c$ depends on $y$. Since $f(1)=1$, for all sufficiently large $\beta$, 
\begin{equation}\label{leveling}
|f(t)- 1|\le e^{-c\beta}
\end{equation}
for $t \in [y,1]$ and a different $c$. 
Also, 
\begin{equation*}
\int_y^1 e^{-\beta V(x)}\,dx = \left(1+ k_2\beta^{-1}+O\left(\beta^{-2}\right)\right)\sqrt{\beta^{-1}}\left(\frac{\pi}{2 V''(1)}\right)^{1/2}\,e^{-\beta V(1)}.
\end{equation*}
where $k_2$ is a $\beta$-independent constant coming from 
the second term in the Laplace expansion. Thus 
\begin{align}\begin{split}\label{part4}
\int_0^1 f(x)e^{-\beta V(x)}\,dx &= O\left(e^{-\beta V(y)}\right)+ \int_y^1 f(x)e^{-\beta V(x)}\,dx \\
&= O\left(e^{-\beta V(y)}\right) + \left(1+ O\left(e^{-c\beta}\right)\right)\int_y^1 e^{-\beta V(x)}\,dx \\
&= O\left(e^{-\beta V(y)}\right) + \left(1+ k_2\beta^{-1}+O\left(\beta^{-2}\right)\right)\sqrt{\beta^{-1}}\left(\frac{\pi}{2 V''(1)}\right)^{1/2}\,e^{-\beta V(1)}\\
&= \left(1+ k_2\beta^{-1}+O\left(\beta^{-2}\right)\right)\sqrt{\beta^{-1}}\left(\frac{\pi}{2 V''(1)}\right)^{1/2}\,e^{-\beta V(1)}.
\end{split}
\end{align}
Using~\eqref{part4} and Lemma~\ref{lemma6} again,
\begin{equation*}
\int_0^1 u(x)e^{-\beta V(x)}\,dx =  \left(1 +  k_2\beta^{-1}+ O\left(\beta^{-2}\right)\right) \sqrt{\beta^{-1}}\left(\frac{\pi}{2 V''(1)}\right)^{1/2}\, e^{-\beta V(1)}.
\end{equation*}
From Remark~\ref{remark2}, we can make an identical argument on $[1,b)$ to get
\begin{equation}\label{integral}
\int_D u(x)e^{-\beta V(x)}\,dx =  \left(1 +  k_2\beta^{-1}+O\left(\beta^{-2}\right)\right) \sqrt{\beta^{-1}}\left(\frac{2\pi}{V''(1)}\right)^{1/2}\, e^{-\beta V(1)}.
\end{equation}
with a different but still $\beta$-independent $k_2$. 
This takes care of the fourth term in the product in~\eqref{C}. 
Observe that in the limit $\beta^{hi},\beta^{lo} \to \infty$, $\beta^{lo}/\beta^{hi} = r$ 
we have:
\begin{align*}
&\frac{1+k_1(\beta^{hi})^{-1} + O((\beta^{hi})^{-2})}{1+k_1(\beta^{lo})^{-1}+O((\beta^{lo})^{-2})} = 1 + O\left(\frac{1}{\beta^{hi}}- \frac{1}{\beta^{lo}}\right)\\
&\frac{1+k_2(\beta^{lo})^{-1} + O((\beta^{lo})^{-2})}{1+k_2(\beta^{hi})^{-1}+O((\beta^{hi})^{-2})} = 1 + O\left(\frac{1}{\beta^{hi}}- \frac{1}{\beta^{lo}}\right).
\end{align*}
Reintroducing 
the superscripts $^{hi}$ and $^{lo}$ and using~\eqref{uprime} and~\eqref{integral} 
in~\eqref{C} now gives 
\begin{equation}\label{main2}
\frac{\lambda^{hi}p_1^{hi}}{\lambda^{lo}p_1^{lo}} = \left(1 + O\left(\frac{1}{\beta^{hi}}- \frac{1}{\beta^{lo}}\right)\right) e^{-(\beta^{hi}-\beta^{lo})(V(0)-V(1))}
\end{equation}
as desired.
\end{proof}

\section{Conclusion}

We have presented a mathematical framework for TAD which is valid in 
any dimension, along with a complete analysis of TAD in one dimension 
under this framework. This framework uses the notion of
quasi-stationary distribution, and is useful in particular to clarify
the immediate equilibration assumption (or no-recrossing assumption) which is underlying the original
TAD algorithm and to understand the extrapolation rule using the
Arrhenius law.
We hope to extend this justification of the extrapolation rule to high 
dimensions, using techniques from~\cite{Tony2}; the analysis 
seems likely to be technically detailed.

We hope that our 
framework for TAD will be useful in cases where the original 
method is not valid. Indeed, we have shown that TAD can be 
implemented wherever accurate estimates for the 
ratios in~\eqref{ratios} are available. This fact is important for transitions which 
pass through degenerate saddle points, in which case a pre-exponential 
factor is needed on the right hand side of~\eqref{main}. 
For example, in one dimension, 
a simple modification of our analysis shows that if we consider 
degenerate critical points on $\partial D$, then a 
factor of the form $(\beta^{hi}/\beta^{lo})^{\alpha}$ 
must be multiplied with the right hand side of~\eqref{main}.

\section*{Acknowledgments}

{\sc D. Aristoff} gratefully acknowledges enlightening discussions with {\sc G. Simpson} 
and {\sc O. Zeitouni}. {\sc D. Aristoff} and {\sc T. Leli\`evre} acknowledge 
fruitful input from {\sc D. Perez} and {\sc A.F. Voter}. Part of this work was completed while 
{\sc T. Leli\`evre} was an Ordway visiting professor at the University of Minnesota.
The work of {\sc D. Aristoff} was supported in part by DOE Award DE-SC0002085.


\begin{thebibliography}{99}
\bibitem{TAD9} {\sc X.M. Bai, A. El-Azab, J. Yu, and T.R. Allen}, 
{\em Migration mechanisms of oxygen interstitial clusters in UO2}, 
J. Phys. Condens. Matter., 25(1) (2013), pp.~015003.


\bibitem{Berglund} {\sc N. Berglund},
{\em Kramers’ law: Validity, derivations and generalisations}, 
Markov Processes Relat. Fields, 19  (2013), pp.~459-490. 

\bibitem{TAD11} {\sc A.V. Bochenkova and V.E. Bochenkov}, 
{\em HArF in Solid Argon Revisited: Transition from Unstable to Stable Configuration}, 
J. Phys. Chem. A, 113(26) (2009), pp.~7654--7659.

\bibitem{Bovier} {\sc A. Bovier, M. Eckhoff, V. Gayrard, and M. Klein},
{\em Metastability in reversible diffusion dynamics I. Sharp asymptotics for capacities and exit times}, {J. Eur. Math.
Soc.}, {6(4)} (2004), pp.~399--424.

\bibitem{TAD6} {\sc M. Cogoni, B.P. Uberuaga, A.F. Voter, and L. Colombo}, 
{\em Diffusion of small self-interstitial clusters in
silicon: temperature-accelerated tight-binding molecular dynamics simulations}, Phys. Rev. B, 71 (2005), 
pp.~121203-1--4.

\bibitem{Day2} {\sc M.V. Day}, {\em On the exponential exit law in the small parameter exit problem}, 
Stochastics, {8} (1983), pp.~297--323. 

\bibitem{Day} {\sc M.V. Day}, {\em Mathematical Approaches to the Problem of Noise-Induced Exit}, 
{Systems and Control: Foundations and Applications}, Springer, 1999, pp.~267--287.

\bibitem{Ofer} {\sc A. Dembo and O. Zeitouni}, {\em Large deviations techniques and applications}, 
Second ed., Springer Verlag, New York, 1998.


\bibitem{Devinatz}  {\sc A. Devinatz and A. Friedman}, 
{\em Asymptotic Behavior of the Principal Eigenfunction for a Singularly 
Perturbed Dirichlet Problem}, {Indiana Univ. Math. J.}, {27} (1978), pp.~143--157.

\bibitem{Friedman}  {\sc A. Friedman} {\em Stochastic differential equations and applications}, 
Vol.~2, Academic Press, New York, 1976.

\bibitem{Hanggi} {\sc P. H\"anggi, P. Talkner, and M. Borkovec}, {\em Reaction-rate theory: fifty years after Kramers}, 
{Rev. Mod. Phys.}, {62(2)} (1990), pp.~251--342.


\bibitem{TAD5} {\sc D.J. Harris, M.Y. Lavrentiev, J.H. Harding, N.L. Allan, and J.A. Purton}, 
{\em Novel exchange mechanisms in the surface diffusion of oxides}, {J. Phys. Condens. Matter}, {16} (2004), 
pp.~187--92.

\bibitem{Helffer} {\sc B. Helffer and F. Nier}, 
{\em Quantitative analysis of metastability in reversible diffusion processes via a Witten complex approach : the case with boundary}, {\em M\' emoires de la SMF}, {105} (2006).


\bibitem{Henkelman} {\sc G. Henkelman, B.P. Uberuaga, and H. J\'onsson},  
{\em A climbing image nudged elastic band method for finding saddle points and minimum energy paths}, 
{J. Chem. Phys.}, {113} (2000), pp.~9901--9904.



\bibitem{Tony} {\sc C. Le Bris, T. Leli\`evre, M. Luskin, and
   D. Perez}, {\em A mathematical formulation  of the parallel replica dynamics}, 
{Monte Carlo Methods Appl.}, {18} (2012), pp.~119--146.

\bibitem{Tony2} {\sc T. Leli\`evre and F. Nier}, {\em Low temperature asymptotics for 
Quasi-Stationary Distribution in a bounded domain}, arXiv:1309.3898.

\bibitem{Menz} {\sc G. Menz and A. Schlichting}, {\em Poincar\'e and logarithmic Sobolev inequalities by decomposition of the energy landscape}, arXiv:1202.1510.

\bibitem{TAD3} {\sc F. Montalenti, M.R. S\o rensen, and A.F. Voter}, 
{\em Closing the gap between experiment and theory: crystal
growth by temperature accelerated dynamics}, {Phys. Rev. Lett.},  87 (2001), 
pp.~126101-1--4.



\bibitem{Voter3} {\sc F. Montalenti and A.F. Voter}  
{\em Exploiting past visits or minimum-barrier knowledge to gain further boost in
the temperature-accelerated dynamics method}, {J. Chem. Phys.}, {116} (2002), 
pp.~4819--4828.

\bibitem{Montalenti2} {\sc F. Montalenti, A.F. Voter, and R. Ferrando},  
{\em Spontaneous atomic shuffle in flat terraces: Ag(100)}, {Phys. Rev. B}, {66} (2002), 
pp.~205404-1--7.

\bibitem{Mousseau} {\sc N. Mousseau, L.K. B\'eland, P. Brommer,  
J.F. Joly, F. El-Mellouhi, E. Machado-Charry, M.C. Marinica, and P. Pochet},  
{\em The Activation-Relaxation Technique: ART Nouveau and Kinetic ART}, 
{Journal of Atomic, Molecular, and Optical Physics}, (2012), pp.~925278.

\bibitem{Nier} {\sc F. Nier}, {\em Boundary conditions and subelliptic estimates for geometric Kramers-Fokker-Planck operators on manifolds with boundaries}, arXiv 1309.5070.

\bibitem{Oksendal} {\sc B. \O ksendal}, {\em Stochastic differential equations: an 
introduction with applications}, Springer Verlag, New York, 2003.

\bibitem{TAD10} {\sc C.W. Pao, T.H. Liu, C.C. Chang, and D.J. Srolovitz}, 
{\em Graphene defect polarity dynamics}, {Carbon}, 50(8) (2012), pp.~2870--2876.

\bibitem{Voter2} {\sc D. Perez, B.P. Uberuaga, Y. Shim, J.G. Amar, 
and A.F. Voter}, {\em Accelerated molecular dynamics methods: introduction and recent developments,} 
{Annual Reports in Computational Chemistry}, {5} (2009), pp.~79--98.

\bibitem{Gideon} {\sc G. Simpson and M. Luskin}, {\em Numerical
    Analysis of Parallel Replica Dynamics}, ESAIM: Mathematical
  Modelling and Numerical Analysis, 47(5) (2013), pp.~1287--1314.

\bibitem{Voter} {\sc M.R. S\o rensen and A.F. Voter}, 
{\em Temperature-accelerated dynamics for 
simulation of infrequent events}, {J. Chem. Phys.}, {112} (2000),
pp.~9599--9606.

\bibitem{TAD4} {\sc J.A. Sprague, F. Montalenti, B.P. Uberuaga, J.D. Kress, and A.F. Voter}, 
{\em Simulation of growth of Cu on Ag(001) at experimental deposition rates}, 
{Phys. Rev. B}, {66} (2000), pp.~205415-1--10.

\bibitem{TAD8} {\sc D.G. Tsalikis, N. Lempesis, G.C. Boulougouris, and 
D.N. Theodorou}, {\em Temperature accelerated dynamics in glass 
forming materials}, {\em J. Chem. Phys. B}, 114 (2010), pp.~7844--7853.


\bibitem{TAD1} {\sc B.P. Uberuaga, R. Smith, A.R. Cleave, G. Henkelman, R.W. Grimes, 
A.F. Voter, and K.E. Sickafus}, 
{\em Dynamical simulations of radiation damage and defect mobility in MgO}, 
{Phys. Rev. B}, {71} (2005), pp.~104102-1.

\bibitem{TAD2} {\sc B.P. Uberuaga, R. Smith, A.R. Cleave, F. Montalenti, G. Henkelman, 
R.W. Grimes, A.F. Voter, and K.E. Sickafus}, 
{\em Structure and mobility of defects formed from collision cascades in MgO}, 
{Phys. Rev. Lett.}, {92} (2004), pp.~115505-4.

\bibitem{TAD7} {\sc B.P. Uberuaga, S.M. Valone, M.I. Baskes, and A.F. Voter},  
{\em Accelerated molecular dynamics study of vacancies in Pu}, {AIP Conf. Proc.}, 673 (2003), 
pp.~213--5.

\bibitem{KMC} {\sc A.F. Voter}, {\em Introduction to the kinetic Monte Carlo method, 
In Radiation Effects in Solids}, Springer, NATO Publishing Unit, Dordrecht, Netherlands, 2006, pp.~1--24.

\bibitem{VoterParRep} {\sc A.F. Voter}, (1998) {\em Parallel replica method for dynamics of infrequent events}, 
{Phys. Rev. B}, {57} (1998), R13985--88.

\bibitem{Voterpriv} {\sc A.F. Voter}, Private communication.
\end{thebibliography}
\end{document}